\documentclass[journal,onecolumn,12pt, peerreview,draftclsnofoot]{IEEEtran}
\usepackage{amsthm}
\usepackage{amssymb}
\usepackage{algorithm}
\usepackage{algorithmic}
\usepackage{amsmath}
\allowdisplaybreaks[4]
\usepackage{amsfonts}
\usepackage{graphicx}
\usepackage{epstopdf}
\usepackage{makeidx}
\usepackage{cite}
\usepackage{float}
\usepackage{epstopdf}
\usepackage{subfigure}
\usepackage{array}
\usepackage{color}
\usepackage{bm}
\usepackage{hyperref}
\usepackage{stfloats}

\newtheorem{theorem}{Theorem}

\IEEEoverridecommandlockouts

\begin{document}

%\font\myfont=cmr12 at 28pt
\title{Performance Analysis of the Full-Duplex Communicating-Radar Convergence System}

\author{Yinghong Guo, Cheng Li, Chaoxian Zhang, Yao Yao,~\IEEEmembership{Senior~Member,~IEEE}, Bin Xia,~\IEEEmembership{Senior~Member,~IEEE}

% <-this % stops a space

%\thanks{This paper was presented in part at the IEEE ICC 2018\cite{8422526}.}
\thanks{Yinghong Guo, C. Li, Y. Yao and B. Xia are with the Department of Electronic Engineering, Shanghai Jiao Tong University (SJTU), Shanghai, 200240, China. Email: \{yinghongguo, lichengg, sandyyao, bxia\}@sjtu.edu.cn.}
\thanks{Chaoxian Zhang is with the School of Information Science and Engineering, Xiamen University Tan Kah Kee College, Xiamen, 363105, China. Email: zhangcx@xujc.com }
% \thanks{Y. Yao is with Huawei Technologies Co., Ltd., Shanghai, 201206, China. Email: yyao@eee.hku.hk.}
}
%\thanks{A. Sabharwal is with the Department of Computer and Engineering, Rice University, Houston, TX, 77005, USA, email: ashu@rice.edu.}

\maketitle
\begin{abstract}
This paper aims to explore the feasibility of the spectrum sharing between the communication and radar system.
\textcolor{black}{We investigate the full-duplex (FD) joint radar and communication multi-antenna system in which a node labeled ComRad with a dual communication and radar capability is communicating with a downlink and an uplink users, as well as detecting the target of interest simultaneously.} Considering a full interference scenario and imperfect channel state information (CSI), the fundamental performance limits of the FD JRC system are analyzed. In particular, we first obtain the downlink rate when the radar signals act as interference.  Then, viewing the uplink channel and radar return channel as a multiple access channel, we propose an alternative successive interference cancellation scheme, based on which the achievable uplink communication rate is obtained. For the radar operation, we first derive the general expression of the estimation rate, which quantifies how much information is obtained about the target in terms of the direction, the range and the velocity. Considering the uniform linear antenna array and linear frequency modulated radar signals, we further obtain the exact closed-form estimation rate. Numerical simulations reveal that the joint manner of the communication and radar operations achieves larger rate regions compared to that of working independently.
\end{abstract}


\section{Introduction}

% The development of the wireless communication technologies has to envisage the challenging problem of explosive data traffic demand \cite{8869705}, which incurs heavy pressure on the spectrum requirement.
%The wireless network capacity is proportional to the system frequency bandwidth, thus, more spectrum is desired.
\textcolor{black}{The development of the wireless communication technologies has to envisage the chxallenging problem of satisfying ever-changing services with explosive data traffic \cite{ZHOU2020253}, which incurs heavy pressure on the spectrum requirement.}
% The international telecommunication union has forecasted that the spectrum deficit of international mobile telecommunications will approach to one thousand megahertz \cite{itu2019wrc}.
\textcolor{black}{The spectrum deficit of international mobile telecommunications to satisfy the ever-increasing demands  becomes a major concern of the industry and academia \cite{8663968}.}
\textcolor{black}{As a promising solution, the joint radar and communication systems have been proposed to share the radar frequency bands with the communication system, by which the spectrum pressure of the communication systems can be alleviated.} On the other hand, the emerging platforms, such as unmanned aerial vehicles (UAVs) and smart cars, require both communication and radar detection operations for safety purpose. In addition, the joint design of the radar and communication could bring advantages of lower hardware cost, saving space and higher energy and spectrum efficiency \cite{Fan2020Survey}.

\subsection{Related Works}
Recently, lots of efforts have been made along the line from independent working to {\color{black}the} joint working of the radar and communication systems. To facilitate the spectrum sharing between the radar and communication system, \cite{mahal2017spectral,biswas2020design,8352726,liu2017robust,8447442,8531782} contributed to the methods designed to alleviate the interference effects.
From the radar performance aspect, \cite{mahal2017spectral} proposed three design methods based on the null-space projection of the radar waveform to mitigate the interference of the radar waveform to the cellular communication system while sharing the same frequency band.
In \cite{biswas2020design}, the coexistence of multiple-input multiple-output (MIMO) cellular system and MIMO radar was studied, where the precoders {\color{black}are} jointly designed for both system to ensure the radar probability of detection and quality of service of cellular users.
 Further contribution \cite{8352726} investigated the joint design of a MIMO radar with co-located antennas and a MIMO communication system aiming at maximizing the signal-to-interference-plus-noise ratio at the radar receiver.
Moreover, \cite{liu2017robust,8447442} studied the beamforming of the downlink communication waveform under the radar detection probability maximization and transmit power minimization criteria, respectively.  The performance limit of the coexistence system was analyzed in \cite{8531782}, where the radar station needs to obtain the decoded communication signals and the communication station needs to obtain the target's estimated information to perform the interference cancellation procedures.

Although the above works have made great progress on the coexistence system design, the radar signals and the communication signals still act as interference and degrade the performance of each other. In addition, since the radar and communication stations are deployed separately and work independently, the exchange of system information, such as the channel state information (CSI) and waveform information, incurs significant system overhead and therefore deteriorates the system performance.

With the emerging of newly developing applications, such as self-driving cars and UAVs, further contributions have considered the dual-function systems where the radar and communication operations are performed simultaneously sharing the same frequency band and infrastructure\cite{Riihonen2021Opt}. While working together, the radar and communication waveform can be jointly designed with the CSI, beamforming schemes, and transmission power shared inherently. Early researches \cite{chiriyath2016inner,chiriyath2017radar,chiriyath2017simultaneous} have exploited the performance limit from the information theory aspect of a single-antenna joint radar and communication system. To unify the radar and communication performance metric, the concept of \emph{estimation information rate} based on the mutual information reduction during estimation procedure has been proposed to evaluate how much information can be obtained about the target within {\color{black}a} unit time. Based on the water-filling algorithm, the optimal bandwidth allocation schemes were also proposed to allocate the whole bandwidth, one part of which is for the dual functions while the other is only for communication or radar operation \cite{chiriyath2016inner,chiriyath2017radar}. 
%Targeting on the dual-functional waveform design, \cite{zhou2019joint,wang2019power,kang2019radar} proposed to use the orthogonal frequency division multiplexing (OFDM) signals acting as the dual functions taking its advantages of robustness against multi-path fading and simple synchronization properties. Based on OFDM waveform, an joint power allocation and subcarrier selection scheme was developed by Shi et al. \cite{shi2020joint}. In addition, \cite{reichardt2012demonstrate,kumari2018ieee802,8845121} have demonstrated the possibility of the dual-function waveform based on the IEEE 802.11p and IEEE 802.11ad protocols, respectively.
More generally, the information theory based capacity-distortion region was derived for joint sensing and communication system where the channel state was estimated at the transmitter by means of generalized feedback \cite{ahmadipour2021informationtheoretic}.
However, with a single antenna, the radar can only obtain the distance and velocity information of targets, whereas the direction information {\color{black}is} ignored. In order to obtain the location and status of targets, both direction information and velocity information are needed leading to the equipment of multiple antennas at the radar node \cite{rong2017multiple,cheng2018outer}.

For the multi-antenna dual-function system, the beamforming schemes have been investigated for the separated scenario and shared scenario \cite{liu2018mumimo}. In the separated scenario, radar and communication use different antennas and independent waveforms. Whereas in the shared scenario, radar and the communication share the same antennas and use the identical waveform, i.e., the communication waveform is used for radar probing simultaneously. Accordingly, an extended Kalman filtering framework was employed to enhance the sensing accuracy of predictive beamforming in vehicle-to-infrastructure(V2I) scenario\cite{liu2020radar}. Targeting on the dual-functional waveform design for the multi-antenna system, \cite{Dokhanchi2019mmWave,wang2019power,xu2021wideband} proposed to use the orthogonal frequency division multiplexing (OFDM) signals acting as the dual functions taking its advantages of robustness against multi-path fading and simple synchronization properties. Furthermore, the novel waveform based on the IEEE 802.11ad protocol was designed for a full-duplex (FD) source transmitter which serves for both radar and communication system \cite{kumari2018ieee802,Kumari2020Adaptive}.
% Moreover, \cite{8828030,8835615} have proposed to enhance the mmWave beam-alignment by utilizing the radar function under the vehicular to vehicular communication scenario. 
% It should be highlighted that in the above works, the superiority of the dual-function radar and communication schemes is unclear due to the absence of fundamental performance limits in the context of multi-antenna dual-function system structures.

In addition, \cite{barneto2021full} has illustrated that simultaneous transmit-and-receive operation is the key enabler for future JRC systems, which leads to communication and radar signals interfering with each other in FD mode. It was claimed that the FD technology could potentially double the spectral efficiency compared with half-duplex counterpart, as it enables simultaneous transmission and reception\cite{barneto2021FDmag}. However, the non-orthogonal operation incurs heavy interference. This motivates us to carry out the study on the feasibility of the FD JRC system from the theoretical aspect. Note that, the above-mentioned works focus on multi-antenna dual-function system {\color{black}operating} in half-duplex mode. However, for FD multi-antenna JRC systems, the potential performance is likely to be enhanced when interference {\color{black}is handled properly}. 
To the best of our knowledge, the fundamental performance analysis for the FD multi-antenna JRC system, where the performance of radar and communication is jointly enhanced, has not been investigated in the literature.

% In addition, the multiple antennas can bring many advantages, such as spatial diversity and higher estimation accuracy for the radar operation, as well as higher transmission and beamforming capability for the communication operation. What's more, with the multiple antennas the interference between the radar and communication can be alleviated by designing the proper beamforming schemes. Although with the above virtues, the performance limit of the multi-antenna joint radar-communication system has not been analyzed theoretically yet to the best of the authors' knowledge.

\subsection{The Contribution of This Work}
In this paper, we investigated the performance bound of a full-duplex multi-antenna joint radar and communication system where a node called ComRad equipped with multiple antennas is performing both the radar and communication operations within the same frequency band.  For the radar operation, the ComRad sends one stream of probing signals and detects the target according to the radar return signals. For the communication operation, the ComRad is sending signals to a downlink communication node (DCN) and receiving signals from an uplink communication node (UCN) simultaneously.
%To facilitate the dual functions, we investigate the performance limit of the radar estimation by the metric of estimation rate about the distance, direction and velocity information of the target. The communication performance bound for both the downlink and uplink under the interfering of the radar signals have been analyzed as well. %
%In this paper, we study the multiple antenna joint radar and communication systems, where the ComRad node is sending radar probing signals to detect the target, transmitting downlink signals to the DCN and receiving uplink signals from the UCN. We consider the full interference scenario that the downlink signals, uplink signals and the radar signals are interfering each other.
By utilizing the term of estimation rate, we unify the radar and communication performance metric which enable us to accurately quantify the trade-off between both functions.
%Considering the practical implementations, the effects of the imperfect channel estimation and imperfect self-interference cancellation are explored as well. 

The main contributions of this paper are illustrated in the following:

\begin{itemize}
\item For the joint radar and communication system, we propose a receiver processing structure, based on which the receiver can decode the communication information and estimate targets' information. 
%To be specific, upon the original received signals, the ComRad first suppress the self-interference and radar return signals \cite{su2018time} since downlink waveform and radar waveform are known, the remained signals are then used for decoding the uplink communication information. Next, with the known uplink waveform, the receiver can subtract the uplink signals from the original received signals and obtain communication interference free radar return signals. Finally, the ComRad can perform estimation process to extract the target's information of interest.
To be specific, the ComRad suppresses the self-interference and radar return signals to decode the uplink communication information, which is then subtracted from the original received signals to extract targets' information of interest. 
This structure captures the convenience of information sharing in a full-duplex system while considering practical implementations, i.e., the imperfect self-interference cancellation and channel estimation.
\item Based on the proposed structure, we study the communication performance based on the maximum ratio transmission (MRT) and maximum ratio combing (MRC) reception beamforming. 
The uplink and downlink communication rate is obtained considering interference from radar return signals, full-duplex self-interference, and radar probing signals.
% Although the self-interference and radar return signals can be suppressed, the uplink signals still suffer from the residual self-interference and radar return signals.
%The best-case power level of the residual radar return signals after interference cancellation is derived. 
In particular, we derive the theoretical limit of the uplink communication rate when radar return is not fully suppressed in the received signal due to estimation error.
% In addition, for the downlink rate, the effect of the radar probing signals acting as interference is taken into account.
Taking into account these practical limitations gives us a thorough evaluation of communication performance under the full-duplex system structure.

%Furthermore, considering the practical limitations, the effects of imperfect channel estimation have been analyzed as well.

\item Further, we analyze the performance of radar operations based on the estimation information rate under the effect of communication interference incurred by imperfect cancellation. To get an accurate status of the target, we focus on the distance, direction, and velocity estimation under general signal waveform expressions. 
The special case of linearly frequency modulated (LFM) waveform with uniform linear antenna array under Gaussian distribution interference is considered, and the exact closed-form estimation rates on distance, direction, and velocity are obtained. % 接下来讲雷达估计速率，包括一般形式，特殊形式，以及各种因素的影响等

\item Finally, we present the numerical results to demonstrate the joint radar and communication system rate regions. 
 Our results reveal the trade-off between radar estimation rate and communication rate, as well as the non-negligible impact of the imperfect channel estimations and interference cancellations. The superiorities of the proposed receiver structure are presented by comparing with the joint rate regions under different receiver structures.
We demonstrate the feasibility of spectrum sharing of the radar and communication working in full-duplex mode. 
\end{itemize}
\subsection{Outline of This Paper}
The remainder of this paper is organized as follows. Section II illustrate the joint system model, including the signal model and channel state information (CSI) requirements. In
Section III, the downlink and uplink communication rates are derived. Section IV presents the radar estimation information rate bounds. Section V demonstrates the numerical simulation results and Section VI concludes this paper.
%We consider the scenario that the radar return signals is affected by the uplink communication signals and the downlink communication signals reflected by the target. Since the downlink waveform is known to the ComRad and uplink waveform is known after uplink decoding, the communication interference can be suppressed. For the communication, the uplink communication signals are interfered by the radar return signals, downlink signals bounced by the target and the self-interference signals. The downlink signals are affected by the uplink co-channel interference and the radar signals reflected by the target.

%Contribution: Explore the feasibility of the spectrum sharing between the radar and communication system.

\section{System Model}
\begin{figure}[ht]
  \centering
  \includegraphics[width=0.6\linewidth]{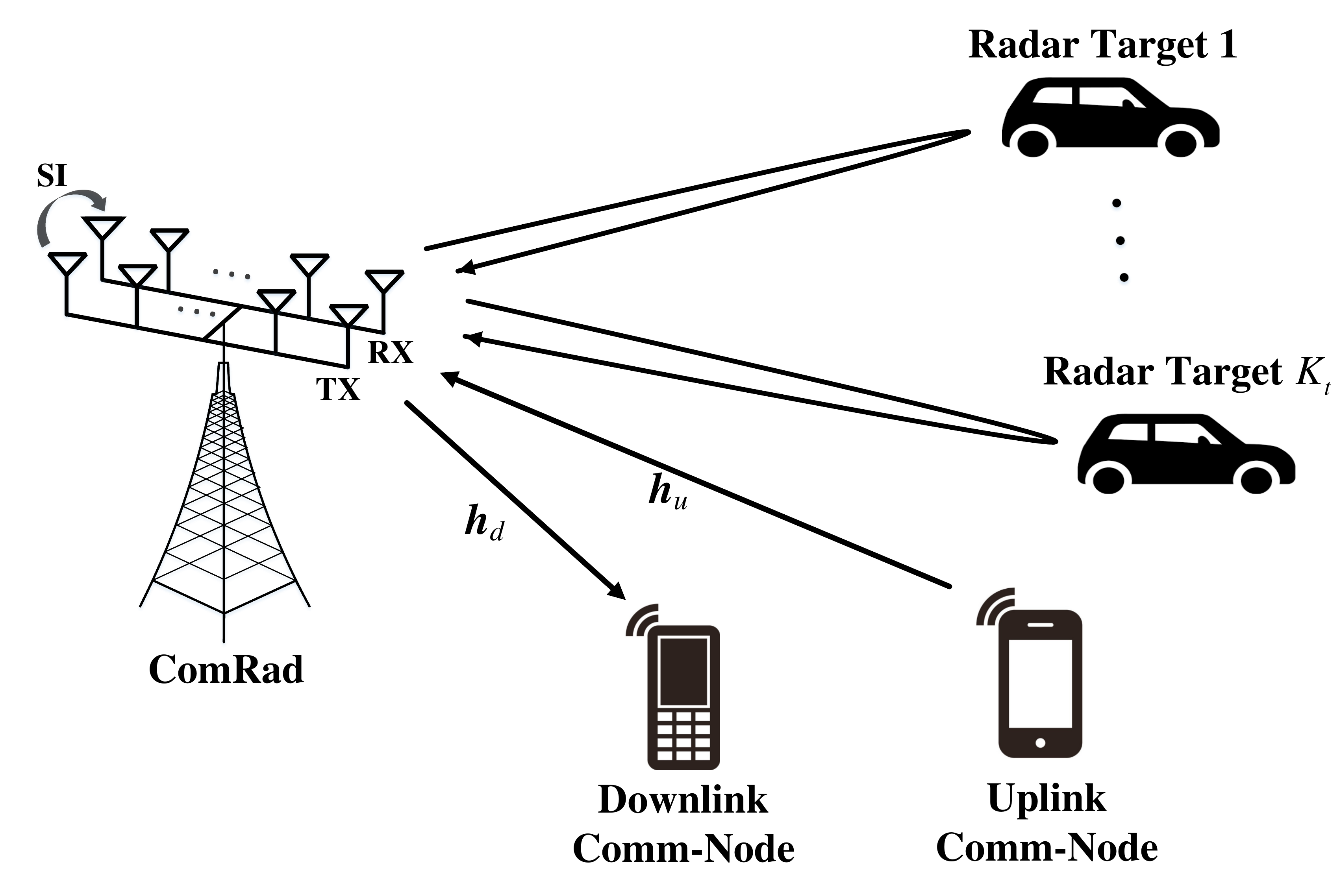}
  \caption{\color{black}A joint full-duplex (FD) communication and radar system model, where a downlink communication node (DCN) and a uplink comm-node (UCN) are communicating with the ComRad, which is detecting the target of interest simultaneously.}\label{system_model}
\end{figure}
The considered multi-antenna FD communication and radar system is depicted in Fig. 1, where the ComRad is receiving uplink signals from the UCN and transmitting downlink signals to the DCN, as well as sending and receiving radar probing signals to detect the target of interest. We assume that the ComRad is equipped with $M$ and $N$ omnidirectional antennas \cite{4350230} for transmitting and receiving, respectively. We consider a simple joint system, where the communication operation, including both the uplink and downlink, and the radar operation are sharing the same frequency band. The communication signals and radar signals are acting as interference to each other. Without loss of generality, we consider the line-of-sight (LoS) channel between the ComRad and the target. Whereas for the communication operation, we consider the Rayleigh distributed non-line-of-sight (NLoS) channel between the communication ends. In this paper, we consider the quasi-static channel, indicating that the communication channel vectors and target parameters are either fixed or change so slightly that they are viewed as constant with each coherent processing interval (CPI) \cite{kumari2018ieee802}.

% {\color{black}It is noted that, in this paper, we aim to prove whether it is feasible for a full-duplex joint radar and communication system. To answer this question, we focus on the simplified case with one downlink and uplink communication node for performance analysis. For a practical system that contains multiple users and each user is capable of both uplink and downlink communication, the multi-user coupling interference will severely affect communication and radar performance. In this case, the mitigation of multi-user interference must be considered which is out of scope of this paper.     }
{\color{black}Noted that we consider a case with only one uplink and downlink communication node to investigate the feasibility of joint radar and communication in full-duplex mode for simplicity and without generality. The derivation and conclusion obtained in single user case are also applicable to multi-user cases when users get access to the network without interference, such as time division multiple access and frequency division multiple access. 
If inter-user interference exists, techniques such as beamforming design, waveform design or interference mitigation methods are in need to alleviate the inter-user interference. However, the methodology is also applicable after considering these interference mitigation methods. As we aim to investigate the trade-off between the radar and communication operations in this paper, the system design to alleviate inter-user interference is out of scope of this paper.}
\subsection{Signal Model}
The signals received at the DCN after co-channel interference cancellation is given by\footnote{The interference from the UCN to the DCN is assumed to be partly cancelled using the methods in \cite{6571320} via wireless side-channels and the residual co-channel interference is modeled as Gaussian noise and its variance is proportional to uplink communication power.} \cite{liu2018mumimo}
\begin{equation}
\bm{y}_{dl}(t) = \bm{h}_{d}^{T}\bm{w}_{d}x_{d}(t)+ \bm{h}^{T}_{d}\bm{s}(t)+z_{d}(t)
\end{equation}
%\begin{align}
%\bar{R} = &\mathbb{E}\Bigg\{\min\Big\{\log_{2}(1+\gamma_{SUM}),\notag\\
%& \min\big[\log_{2}(1+\gamma_{AR}),\log_{2}(1+\gamma_{RB})\big]\notag\\
%& +\min\big[\log_{2}(1+\gamma_{BR}),\log_{2}(1+\gamma_{BR})\big]\Big\}\Bigg\}
%\end{align}
where $x_{d}(t)$ and $\bm{s}(t)$ denotes the downlink signal and radar probing signals with transmission power $P_{dl}$ and $P_{rad}$, respectively. $\bm{w}_{d}$ denotes the downlink beamforming vector. In this paper, we consider the maximum ratio transmission (MRT). $\bm{h}_{d}$ stands for the downlink frequency-flat independent and identically distributed (i.i.d.) Rayleigh fading channel vector with zero mean and covariance matrix $\Omega_{dl}\bm{I}_{M}$, and finally $z_{d}(t)$ is the Gaussian noise with zero mean and variance $\sigma_{z}^2+K_{co}P_{ul}$ including thermal noise and residual co-channel noise, where $K_{co}$ and $P_{ul}$ denote the co-channel interference suppression level and uplink communication power, respectively.

The uplink signal and the radar return signal at the ComRad receiver are given by
\begin{equation}\label{signal-com-rad}
\bm{y}_{\text{\emph{com-rad}}}(t)= \bm{y}_{rad}(t)+\bm{y}_{ul}(t)+\bm{y}_{self}(t)+\bm{y}_{boun}(t)+\bm{z}_{r}(t),
\end{equation}
where $\bm{y}_{rad}(t)$ and $\bm{y}_{ul}(t)=\bm{h}_{u}x_{u}(t)$ denotes the echo of radar probing signal and received uplink signals, respectively. $x_{u}(t)$ stands for the transmitted uplink signal and $\bm{h}_{u}$ represents the uplink frequency-flat i.i.d. Rayleigh fading channel vector with zero mean and covariance matrix $\Omega_{ul}\bm{I}_{N}$.
%$\bm{y}_{coch}(t)$ denotes the communication interference signals from other nodes in the network.
\textcolor{black}{$\bm{y}_{self}(t)$ and $\bm{y}_{boun}(t)$ denote the self-interference incurred by the FD mode and the downlink communication signals reflected by the targets, respectively.} And finally, $\bm{z}_{r}(t)$ is the Gaussian thermal noise vector at the ComRad with zero mean and variance $\sigma_{0}^{2}$ of each entry.
%  \textcolor{black}{Since $\bm{y}_{self}(t)$ and $\bm{y}_{boun}(t)$ are interference to ComRad which contains no useful information, we do not focus on the specific expression of this interference, but are only concern about the expression of the residual interference after cancellation process.}

The radar return signal reflected by targets of interest can be further expressed by
\begin{equation}
\bm{y}_{rad}(t)= \sum_{k=1}^{K_{t}}\alpha_{k}\bm{a}(\theta_{k})\sum_{i=1}^{M}s_{i}(t-\tau_{k})\exp(j\omega_{k}t),
\end{equation}
where $K_{t}$ represents the number of the far apart targets meaning they can be estimated independently \cite{chiriyath2015joint,chiriyath2016inner}. $\alpha_{k}$ denotes the complex combined path loss and target reflection factor \cite{chiriyath2016inner}. $\bm{a}(\theta_{k})$ denotes the antenna steering vector at the direction $\theta_{k}$. $\exp(j\omega_{k}t)$ characterizes the Doppler frequency shift effect. $s_{i}(t)$ denotes the radar probing signals transmitted by the $i$-th antenna element. Here, we assume that the radar waveforms transmitted by different antenna elements are orthogonal to each other, i.e., $Cov[\bm{s}(t)]=\frac{P_{rad}}{M}\bm{I}_{M}$. In addition, by assuming that the transmit antenna elements meets the orthogonality conditions \cite{4156404}, the received signals can be expressed as the sum of different waveforms. $\tau_{k}$ characterizes the transmission time delay between the ComRad and the $k$-th target. It is noted that the parameters $\theta_{k}$, $\tau_{k}$ and $\omega_{k}$ of the targets are of interest to be estimated in this paper. 

{\color{black}The self-interference signal and downlink signals reflected by the targets can be expressed as
\begin{align}
  \bm{y}_{self}(t) &= {\mathbf{H}}_{self}\left[\bm{w}_{d}x_{d}(t)+ \bm{s}(t)\right]\\
  \bm{y}_{boun}(t) &= {\mathbf{H}}_{boun}\bm{w}_{d}x_{d}(t)
\end{align}
where  ${\mathbf{H}}_{self},{\mathbf{H}}_{boun} \in \mathbb{C}^{N \times M}$ denote the channel matrix of self-interference channel and downlink communication bouncing channel, respectively.

For clarity, although both interference are transmitted from the Tx antennas to the Rx antennas, they are fundamentally different. The bouncing interference is with a large delay which is comparable at symbol level while the delay of self-interference is negligible. Thus, they have to be cancelled separately with different methods.

}

\subsection{CSI Requirements}
It is noted that in the considered joint radar and communication system, there are multiple interference links described in the following.

 \begin{itemize}
   \item  \emph{Communication downlink:} \textcolor{black}{The downlink communication signals} are interfered by 1) the omnidirectional radar probing signals and 2) the uplink communication signals.
   \item \emph{Communication uplink:} \textcolor{black}{The uplink communication signals} are interfered by 1) the radar return signals and 2) self-interference incurred by the FD mode.
   \item \emph{Radar return signals:} The radar return signals are interfered by 1) the uplink communication signals from the UCN and 2) the self-interference incurred by the FD mode.
 \end{itemize}

To correctly decode the desired communication signals and improve the radar estimation accuracy, the CSI of $\bm{h}_{d}$ and $\bm{h}_{u}$ are needed. The estimated $\bar{\bm{h}}_{d}$ and $\bar{\bm{h}}_{u}$ are described as
\begin{equation}
\bm{h}_{d}=\rho_{1}\bar{\bm{h}}_{d}+(\sqrt{1-\rho^{2}_{1}})\bm{\varepsilon}_{1},
\end{equation}
\begin{equation}
\bm{h}_{u}=\rho_{2}\bar{\bm{h}}_{u}+(\sqrt{1-\rho^{2}_{2}})\bm{\varepsilon}_{2},
\end{equation}
where $\rho_{1} \in (0,1)$ and $\rho_{2} \in (0,1)$ denote the correlation coefficients between the true value and the estimated value of the CSI \cite{4801449}. The entries of $\bar{\bm{h}}_{d}$ and $\bar{\bm{h}}_{u}$ are complex Gaussian random variables with zero mean and variances $\Omega_{dl}$ and $\Omega_{ul}$, respectively. $\bm{\varepsilon}_{1}$ and $\bm{\varepsilon}_{2}$ denote the independent complex Gaussian estimation errors with zero mean and variances $\Omega_{dl}$ and $\Omega_{ul}$, respectively.

\section{Performance of the Communication Operations}
In this section, we first derive the downlink communication rate. For the uplink, we propose an alternating-SIC scheme, based on which we then obtain the uplink communication rate.

\subsection{Downlink Communication Rate}
Considering the imperfect channel estimation, the received signals at the DCN are re-written as
\begin{multline}
y_{dl}(t) =  \rho_{1}\bm{w}_{d}^{T}\bar{\bm{h}}_{d}x_{d}(t)+\rho_{1}\bar{\bm{h}}^{\dag}_{d}\bm{s}(t)+(\sqrt{1-\rho^{2}_{1}})\bm{w}_{d}^{T}\bm{\varepsilon}_{1}x_{d}(t)+(\sqrt{1-\rho^{2}_{1}})\bm{\varepsilon}^{T}_{1}\bm{s}(t)+z_{d}(t).
\end{multline}

Since the radar probing signals are deterministic sequence and the $\bar{\bm{h}}_{d}$ is known, the DCN can eliminate the  radar interference $\rho_{1}\bar{\bm{h}}_{d}\bm{s}(t)$ prior to the communication decoding. For the downlink beamforming, we adopt the MRT beamformer. Hence, we have $\bm{w}_{d}^{T}=\frac{\bar{\bm{h}}^{\dag}_{d}}{||\bar{\bm{h}}^{\dag}_{d}||_{2}}$. The remained signals are given by
\begin{align}\label{comm-decoding}
\!\!\!\!\hat{y}_{dl}(t) =& \rho_{1}\frac{\bar{\bm{h}}^{\dag}_{d}}{||\bar{\bm{h}}^{\dag}_{d}||_{2}}\bar{\bm{h}}_{d}x_{d}(t)+(\sqrt{1-\rho^{2}_{1}})\frac{\bar{\bm{h}}^{\dag}_{d}}{||\bar{\bm{h}}^{\dag}_{d}||_{2}}\bm{\varepsilon}_{1}x_{d}(t)+(\sqrt{1-\rho^{2}_{1}})\bm{\varepsilon}^{T}_{1}\bm{s}(t)+z_{d}(t).
\end{align}
\textcolor{black}{Based on the signal expressed in (\ref{comm-decoding}), the downlink communication rate is given by}

\begin{align}
R_{dl}=&f_{B}\mathbb{E}\left\{\log_2\left(1+\frac{\rho_{1}^{2}\sum_{i=1}^{M}|h_{d}(i)|^2P_{dl}}{I_{dl}^{dl}+I_{dl}^{rad}+I_{dl}^{ul}+\sigma_{z}^{2}}\right)\right\}\notag\\
\overset{(a)}\le&f_{B}\log_2\left\{(1+\mathbb{E}\left\{\frac{\rho_{1}^{2}\sum_{i=1}^{M}|h_{d}(i)|^2P_{dl}}{I_{dl}^{dl}+I_{dl}^{rad}+I_{dl}^{ul}+\sigma_{z}^{2}}\right\}\right)\notag\\
=&f_{B}\log_{2}\left(1+\frac{\rho_{1}^{2}MP_{dl}\Omega_{dl}}{I_{dl}^{dl}+I_{dl}^{rad}+I_{dl}^{ul}+\sigma_{z}^{2}}\right)
\end{align}

where (a) is achieved by using the Jensen's inequality \cite{cover2012elements}. $I_{dl}^{dl}=(1-\rho_{1}^{2})\Omega_{dl}P_{dl}$ ,$I_{dl}^{rad}=(1-\rho_{1}^{2})\Omega_{dl}P_{rad}$ and $I_{dl}^{ul}=K_{co}P_{ul}$ characterize the residual downlink signals, radar signals and uplink signals  acting as interference to the downlink decoding, respectively, due to the imperfect channel estimation and interference cancellation.

\subsection{Uplink Communication Rate}
From the expression (\ref{signal-com-rad}), we note that the radar return signals and uplink communication signals are received simultaneously. According to \cite{chiriyath2016inner}, the radar signals and the uplink signals can be viewed as MAC signals. However, different from the traditional MAC, where signals of all the uplink streams are unknown, in the communication-radar MAC, the transmitted radar waveform is known at the receiver. Hence, to improve the uplink decoding performance, an SIC scheme is proposed and shown in the following diagram.

\begin{figure}[ht]
  \centering
  \includegraphics[width=0.8\linewidth]{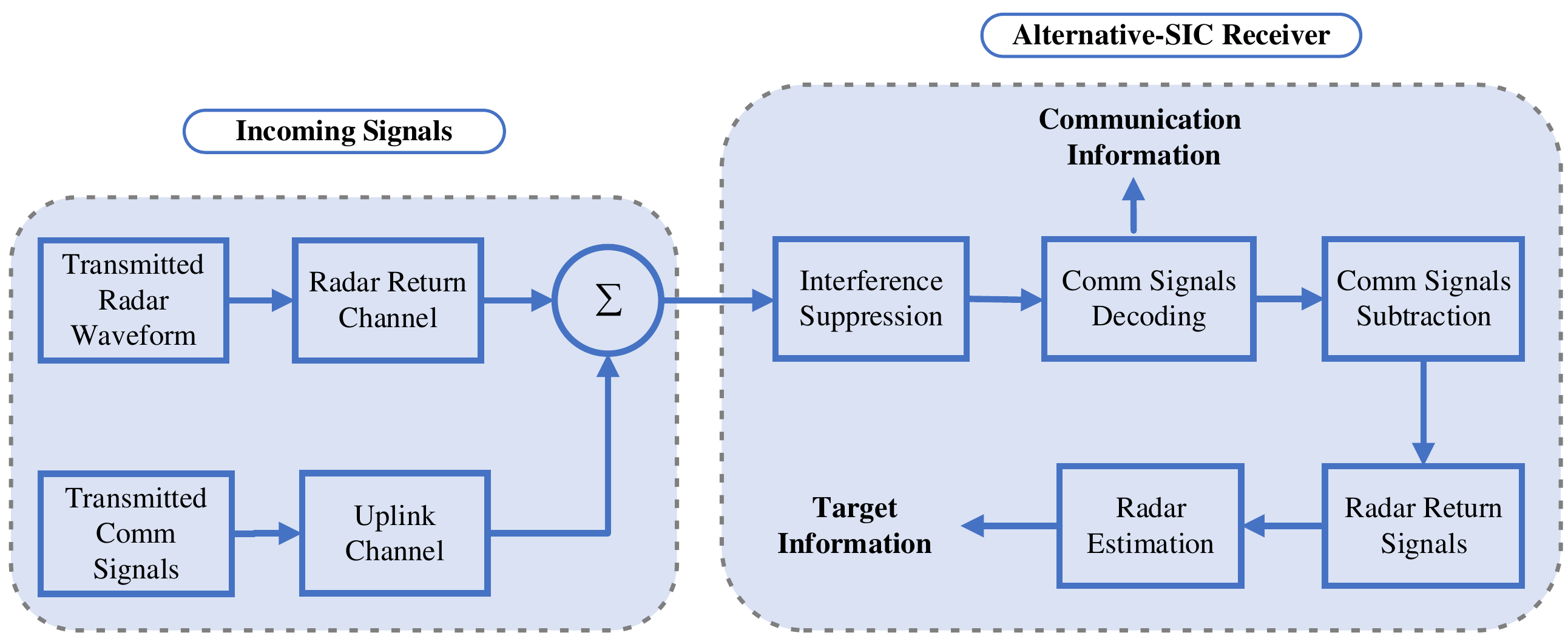}
  \caption{The proposed alternative-SIC scheme for the ComRad receiver.}\label{iterative-SIC}
\end{figure}

The key steps of the proposed alternative-SIC scheme are explained as
\begin{enumerate}
\item \emph{Interference suppression}: In this step, the ComRad first suppress the self-interference $\bm{y}_{self}(t)$ incurred by the FD mode \cite{duarte2012experiment} and the return of the downlink signals $\bm{y}_{boun}(t)$ reflected by the target \cite{wang2014radar}, since these two parts degrades both the radar and communication performance. The remained signals are denoted by $\bm{y}_{\emph{com-rad}}'(t)$. Then, the ComRad suppress the radar return signals $\bm{y}_{rad}(t)$ from $\bm{y}_{\emph{com-rad}}'(t)$ \cite{su2018time}. The remained signals are denoted by $\hat{\bm{y}}_{ul}(t)$.
\item \emph{Communication signals decoding}: Based on $\hat{\bm{y}}_{ul}(t)$, we perform the MRC receiver beamforming and we have $\hat{y}'_{ul}(t)$, which is used for uplink communication signals decoding. In this step, we consider the perfect decoding.
\item \emph{Communication signals subtraction}: After the desired uplink signals have been extracted, the ComRad could subtract $\bm{y}_{ul}(t)$ from $\bm{y}_{\emph{com-rad}}'(t)$. The remained signal is denoted by $\hat{\bm{y}}_{rad}(t)$.
\item \emph{Radar estimation}: After the uplink signals subtraction, the ComRad performs the radar estimation using $\hat{\bm{y}}_{rad}(t)$ and extracts the target's information.
\end{enumerate}

After the interference cancellation step,
\begin{align}
\hat{\bm{y}}_{ul}(t)=&(\rho_{1}\bar{\bm{h}}_{u}+(\sqrt{1-\rho_{1}^{2}})\bm{\varepsilon}_{1})x_{u}(t)+\bar{\bm{y}}_{boun}(t)+\bar{\bm{y}}_{self}(t)+\bar{\bm{y}}_{rad}(t)+\bm{z}_{r}(t),
\end{align}
where $\bar{\bm{y}}_{self}(t)$, $\bar{\bm{y}}_{boun}(t)$, and $\bar{\bm{y}}_{rad}(t)$ denote the residual self-interference (RSI), residual downlink communication signals reflected by the target, and residual radar return signals, respectively.

{\color{black}As the transmitted radar waveform and downlink communication signal is already known to the ComRad, with the estimated self-interference channel ${\bf{H}}_{self}$ and communication bouncing channel ${\bf{H}}_{boun}$, the ComRad can subtract the self-interference and bouncing communication signal. The residual interference after cancellation of ComRad can be expressed as
\begin{align}
  \bar{\bm{y}}_{self}(t) &= {\mathbf{H}}_{self}{\tilde{\bm{x}}}_{self}(t)\\
  \bar{\bm{y}}_{boun}(t) &= {\mathbf{H}}_{boun}{\tilde{{\bm{x}}}_{boun}(t)}
\end{align}
where ${\tilde{\bm{x}}}_{self}(t)$ and ${\tilde{\bm{x}}}_{boun}(t)$ represent the residual self-interference signals and and bouncing signals after cancellation at the ComRad.
}

\textcolor{black}{ According to the experiment characterization of the residual self-interference after interference suppression \cite{duarte2012experiment}, $\bar{\bm{y}}_{self}$ can be modeled as additional Gaussian noise and subject to $\bar{\bm{y}}_{self}\sim \mathcal{CN}(0,K_{self}(P_{dl}+P_{rad})\mathbf{I}_{N})$, i.e. the variance of the noise is proportional to the sum of downlink transmission power $P_{dl}$ and radar transmission power $P_{rad}$, where $K_{self}$ denotes the self-interference cancellation capability \cite{7801128}. Similarly, as the residual bouncing interference $\bar{\bm{y}}_{boun}$ inherits the statistical randomness of communication bouncing signal that carries downlink communication symbols and unknown radar parameters, we  modeled is as Gaussian \cite{wang2014radar} with variance $K_{boun}P_{dl}\mathbf{I}_{N}$, which have shown a good match to the statistical distribution of the residual interference from the simulation\cite{wang2014radar}. Here, $K_{boun}$ characterizes the bouncing interference suppression level. }

{\color{black}We note that $K_{self}$ and $K_{boun}$ are empirical values between one and zero, which are dependent on the RSI suppression schemes and downlink bouncing signal cancellation schemes used by the ComRad, respectively. The interference cancellation scheme with higher suppression level results in smaller $K_{self}$ and $K_{boun}$.}

{\color{black}In the following, we derive the theoretical limit on the power of residual radar return signal when radar parameter is estimated.} By defining $\bm{\hat{h}}_{rad,k}$ as the estimated radar channel between the $k$th target and the ComRad , and $S_{k}(t-\hat{\tau},\hat{\omega}_k)=\sum_{i=1}^{M}s_{i}(t-\hat{\tau}_{k})\exp(j\hat{\omega}_{k}t)$ as the estimated signal reflected by the $k$th target , the residual radar return signals can be expressed as
\begin{align}
  \bar{\bm{y}}_{rad}(t) &= \sum_{k=1}^{K_{t}}\alpha_{k}\bm{a}(\theta_{k})\sum_{i=1}^{M}s_{i}(t-\tau_{k})\exp(j\omega_{k}t)-\sum_{k=1}^{K_{t}}\bm{\hat{h}}_{rad,k}S_{k}(t-\hat{\tau},\hat{\omega}_k)
\end{align}

In order to evaluate the achievable performance limit of the radar return signal suppression process, we apply the first order Taylor series expansion to approximate the residual signal which can be rewritten into
\begin{align}\label{appro_level}
  &\bar{\bm{y}}_{rad}(t) \approx \sum_{k=1}^{K_{t}}\bm{h}_{rad,k} \left( (\tau-\hat{\tau}_k) \frac{\partial S_{k}(t-\hat{\tau},\hat{\omega}_k)}{\partial \tau_k} + (\omega-\hat{\omega}_k) \frac{\partial S_{k}(t-\hat{\tau},\hat{\omega}_k)}{\partial \omega_k}\right) +\sum_{k=1}^{K_{t}}\bm{\mu}_{k}S_{k}(t-\hat{\tau},\hat{\omega}_k)
\end{align}

% \begin{align}
%   \bar{\bm{y}}_{rad}(t) &= \sum_{k=1}^{K_{t}}\bm{h}_{rad,k}\left( S_{k}(t-\tau,\omega_k)-S_{k}(t-\hat{\tau},\hat{\omega}_k)  \right)+\sum_{k=1}^{K_{t}}\bm{\mu}_{k}S_{k}(t-\hat{\tau},\hat{\omega}_k)\notag\\
%   &\approx \sum_{k=1}^{K_{t}}\bm{h}_{rad,k} \left( (\tau-\hat{\tau}_k) \frac{\partial S_{k}(t-\hat{\tau},\hat{\omega}_k)}{\partial \tau_k} + (\omega-\hat{\omega}_k) \frac{\partial S_{k}(t-\hat{\tau},\hat{\omega}_k)}{\partial \omega_k}\right) \notag\\
%   & \quad +\sum_{k=1}^{K_{t}}\bm{\mu}_{k}S_{k}(t-\hat{\tau},\hat{\omega}_k)
% \end{align}

We have $\bm{h}_{rad,k}=\bm{\hat{h}}_{rad,k}+\bm{\mu}_{k}$, where $\bm{\mu}_{k} \in \mathbb{C}^{M\times 1}$ denotes the unbiased radar channel estimator error matrix whose best-case variance is given by the Cram$\acute{\textup{e}}$r-Rao Lower Bound (CRLB). The variance of $i$-th element of $\bm{\mu}_{k}$ can be expressed as\cite{bliss2013adaptive}
\begin{equation}
  \mathbb{E}\left[ \mu_{k,i}\right] \geq \frac{M\sigma_{0}^{2}}{\delta T_{c} f_{B} P_{rad}}
\end{equation}
where $\delta$, $T_{c}$ and $f_{B}$ denote duty factor, coherent processing interval and bandwidth of the radar signal respectively. Then, we obtain $I_{ul}^{rad}$ as the impact of the interfering residual radar signal after suppression to the uplink signal at the top of next page in (\ref{residual-radar}).
\begin{figure*}
  \vspace{-0.3in}
  \begin{align}\label{residual-radar}
    P_{I} &= \mathbb{E}\left[ \bar{\bm{y}}_{rad}(t)^{H}\bar{\bm{y}}_{rad}(t) \right] = \sum_{k=1}^{K_{t}} \left\{ \alpha_{k}^2\bm{a}^H(\theta_{k})\bm{a}(\theta_{k}) \mathbb{E}\left[ (\tau-\hat{\tau}_k)^2 \left|\frac{\partial S_{k}}{\partial\tau_k}\right|^2 +(\omega-\hat{\omega}_k)^2 \left|\frac{\partial S_{k}}{\partial\omega_k}\right|^2\right] \right\} \notag\\
    &+ \sum_{k=1}^{K_{t}} \left\{ \alpha_{k}^2\bm{a}^H(\theta_{k})\bm{a}(\theta_{k})\mathbb{E}\left[2(\tau-\hat{\tau}_k)(\omega-\hat{\omega}_k)\left|\frac{\partial S_{k}}{\partial\tau_k}\frac{\partial S_{k}}{\partial\omega_k}\right|\right]  \right\}\notag+\frac{K_t\sigma_{0}^{2}}{\delta T_{c} f_{B}} \notag\\
    &\overset{(a)}\geq \sum_{k=1}^{K_{t}} \!\!\left\{ MP_{rad} \!\!\left( 4\pi^{2} B_{rms}^{2} \bm{C}(\tau_k,\tau_k) \!+\!\mathbb{E}\left[ t^2S_{k}^2 \right] \bm{C}(\omega_k,\omega_k) \!+\!\mathbb{E}\left[\frac{tS_{k}\partial S_{k}}{\partial\tau_k}\right] \bm{C}(\tau_k,\omega_k)\right)\!\! \right\}\!\! +\frac{K_t\sigma_{0}^{2}}{\delta T_{c} f_{B}}
  \end{align}
  \hrulefill
\end{figure*}
% \begin{align}\label{residual-radar}
%   I_{ul}^{rad} &= \mathbb{E}\left[ \bar{\bm{y}}_{rad}(t)^{H}\bar{\bm{y}}_{rad}(t) \right] \notag\\
%   &= \sum_{k=1}^{K_{t}} \left\{ \alpha_{k}^2\bm{a}^H(\theta_{k})\bm{a}(\theta_{k}) \mathbb{E}\left[ (\tau-\hat{\tau}_k)^2 \left|\frac{\partial S_{k}}{\partial\tau_k}\right|^2 +(\omega-\hat{\omega}_k)^2 \left|\frac{\partial S_{k}}{\partial\omega_k}\right|^2\right] \right\}\notag\\
%   &\quad + \sum_{k=1}^{K_{t}} \left\{ \alpha_{k}^2\bm{a}^H(\theta_{k})\bm{a}(\theta_{k})\mathbb{E}\left[2(\tau-\hat{\tau}_k)(\omega-\hat{\omega}_k)\left|\frac{\partial S_{k}}{\partial\tau_k}\frac{\partial S_{k}}{\partial\omega_k}\right|\right]  \right\}\notag+\frac{K_t\sigma_{0}^{2}}{\delta T_{c} f_{B}} \notag\\
%   &\overset{(a)}\geq \sum_{k=1}^{K_{t}} \left\{ MP_{rad} \left( 4\pi^{2} B_{rms}^{2} \bm{C}(\tau_k,\tau_k) +\mathbb{E}\left[ t^2S_{k}^2 \right] \bm{C}(\omega_k,\omega_k) +\mathbb{E}\left[\frac{tS_{k}\partial S_{k}}{\partial\tau_k}\right] \bm{C}(\tau_k,\omega_k)\right) \right\} \notag\\
%   &\quad +\frac{K_t\sigma_{0}^{2}}{\delta T_{c} f_{B}}
% \end{align}

{\color{black}where (a) is achieved since the variance of the unbiased estimated parameter of the $k$th target $\mathbb{E}[(\tau-\bar{\tau})^2]$, $\mathbb{E}[(\omega-\bar{\omega})^2]$ and $\mathbb{E}[(\tau-\bar{\tau})(\omega-\bar{\omega})]$ is bounded by $\bm{C}(\tau_k,\tau_k)$, $\bm{C}(\omega_k,\omega_k)$ and $\bm{C}(\tau_k,\omega_k)$ which is the Cram$\acute{\textup{e}}$r-Rao bound (CRB) defined in Theorem 1 in the next section.} We further define $P_{I}=K_{rad}P_{rad}$ indicating that the residual radar signal is proportional to radar power by a factor $K_{rad}$, we obtain that the theoretical limit on suppression level of the radar return signal
\begin{align}\label{radar_suppresion_level}
  K_{rad}^{*} &= M\sum_{k=1}^{K_{t}} \left\{ \left( 4\pi^{2} B_{rms}^{2} \bm{C}(\tau_k,\tau_k) +\mathbb{E}\left[ t^2S_{k}^2 \right] \bm{C}(\omega_k,\omega_k) \right.\right.\notag\\
  &\left.\left.+\mathbb{E}\left[\frac{tS_{k}\partial S_{k}}{\partial\tau_k}\right] \bm{C}(\tau_k,\omega_k)\right) \right\} +\frac{K_t\sigma_{0}^{2}}{\delta T_{c} f_{B}P_{rad}}
\end{align}
% \begin{align}
%   K_{rad}^{*} &= M\sum_{k=1}^{K_{t}} \left\{ \left( 4\pi^{2} B_{rms}^{2} \bm{C}(\tau_k,\tau_k) +\mathbb{E}\left[ t^2S_{k}^2 \right] \bm{C}(\omega_k,\omega_k)+\mathbb{E}\left[\frac{tS_{k}\partial S_{k}}{\partial\tau_k}\right] \bm{C}(\tau_k,\omega_k)\right) \right\} \notag \\
%   &~~ +\frac{K_t\sigma_{0}^{2}}{\delta T_{c} f_{B}P_{rad}}
% \end{align}

It is noted that the theoretical limit on radar signal suppression level $K_{rad}^{*}$ is dependent on the choice of radar signal type $\bm{s}(t)$. 

Subsequently, with the MRC beamformer, we have
\begin{align}\label{signal-up}
\hat{y}'_{ul} =&\rho_{1}\frac{\bar{\bm{h}}^{\dag}_{u}}{||\bar{\bm{h}}_{u}||_{2}}\bar{\bm{h}}_{u}x_{u}(t)+(\sqrt{1-\rho_{1}^{2}})\frac{\bar{\bm{h}}^{\dag}_{u}}{||\bar{\bm{h}}_{u}||_{2}}\bm{\varepsilon}_{1}x_{u}(t)\notag\\
&+\frac{\bar{\bm{h}}^{\dag}_{u}}{||\bar{\bm{h}}_{u}||_{2}}\left[\bar{\bm{y}}_{rad}(t)+\bar{\bm{y}}_{boun}(t)+\bar{\bm{y}}_{self}(t)+\bm{z}(t)\right].
\end{align}

Based on the signals expressed in (\ref{signal-up}), the uplink communication rate is given by
\begin{align}\label{uplink-rate}
  R_{ul} =&\mathbb{E}\left\{f_{B}\log_{2}\left(1+\frac{\rho^{2}_{1}\sum_{i=1}^{N}|\bar{h}_{u}(i)|^2P_{ul}}{I_{ul}^{ul}+I_{ul}^{dl}+I_{ul}^{rad}+\sigma^{2}_{z}}\right)\right\} \notag\\
  \overset{(b)}\le& f_{B}\log_{2}\left(1+\mathbb{E}\left\{\frac{\rho^{2}_{1}\sum_{i=1}^{N}|\bar{h}_{u}(i)|^2P_{ul}}{I_{ul}^{ul}+I_{ul}^{dl}+I_{ul}^{rad}+\sigma^{2}_{z}}\right\}\right)\notag\\
  =&f_{B}\log_{2}\left(1+\frac{\rho^{2}_{1}NP_{ul}\Omega_{ul}}{I_{ul}^{ul}+I_{ul}^{dl}+I_{ul}^{rad}+\sigma^{2}_{z}}\right),
\end{align}
where (b) is achieved by using the Jensen's inequality \cite{cover2012elements}. $I_{ul}^{ul}=(1-\rho^{2}_{1})\Omega_{ul}P_{ul}$, $I_{ul}^{dl}=(K_{self}+K_{boun})P_{dl}$ and $I_{ul}^{rad}=(K_{self}+K_{rad})P_{rad}$ characterize the impacts of the residual co-channel uplink signals, downlink signals and radar signal acting as interference to the uplink decoding, respectively. $K_{rad}$, $K_{self}$ and $K_{boun}$ characterize suppression level of the radar signals, self-interference and bounces of the downlink signals, respectively. 

It is noted that the residual interference after suppression are modeled as complex Gaussian random variables due to the imperfect cancellation. The corresponding uplink then given by substituting (\ref{residual-radar}) into (\ref{uplink-rate}), which is the achievable uplink communication rate of the proposed alternative-SIC scheme.

\section{Performance of the Radar Operation}
To unify the performance metric of the joint radar and communication system, the estimation rate, which is analogous to the communication rate for communication operation, is derived to evaluate the performance of the radar operations\cite{chiriyath2015joint}. Different from the traditional performance metrics of radar system, i.e., the mean-square error (MSE) \cite{4156404}, the estimation rate is a measure of how much information about the target is obtained within unit time period. Since the radar targets are assumed to be far apart, the system sum estimation rate can be viewed as the sum of each target's estimation rate \cite{chiriyath2016inner}. In the following, we derive the estimation rate for the $k$-th target.

The estimation rate can be calculated by
\begin{equation}
R_{est}=\frac{1}{T_{c}}\log(1+\frac{\sigma_{dyn}^{2}}{\text{CRB}_{est}})
\end{equation}
where $T_{c}$ is the duration of each CPI. $\sigma_{dyn}^{2}$ is the variance of the target's unknown dynamic process under Gaussian assumption. $\text{CRB}_{est}$ is the Cram$\acute{\textup{e}}$r-Rao bound, which is the lower bound of the MSE of all unbiased estimators \cite{4156404}, of the estimated parameter.

After the uplink signals being extracted, the ComRad can subtract the received uplink signal $\bm{y}_{ul}$ from $\bm{y}'_{com-rad}(t)$.
\textcolor{black}{Hence, $\hat{\bm{y}}_{rad}(t)$ is given by}
\begin{align}
\!\!\!\!\!\!\!\!\hat{\bm{y}}_{rad}(t)=&\alpha_{k}\bm{a}(\theta_{k})\exp(j\omega_{k}t)\sum_{i=1}^{M}s_{i}(t-\tau_{k})\notag\\
&+\!\sqrt{1-\rho^2_{1}}\bm{\varepsilon}_{1}x_{u}(t)+\bar{\bm{y}}_{self}(t)+\bar{\bm{y}}_{boun}(t)+\bm{z}(t).
\end{align}

After sampling and stacking all the samples into one vector, we have
\begin{align}\label{rad-sample}
\hat{\bm{y}}_{rad}=&\left[\hat{\bm{y}}_{rad}[1],...,\hat{\bm{y}}_{rad}[L]\right]\notag\\
=&\hat{\bm{y}}'_{rad}+\hat{\bm{z}}\notag\\
=&\alpha_{k} \cdot \bm{y}'_{rad}\otimes \bm{a}(\theta_{k})+\hat{\bm{z}},
\end{align}
where $\hat{\bm{y}}_{rad}[l]$ denotes the \emph{l}-th sample of $\hat{\bm{y}}_{rad}(t)$ and it is given by
\begin{align}
\hat{\bm{y}}_{rad}[l]=&\alpha_{k}\bm{a}(\theta_{k})\exp(j\omega_{k}l\Delta t)\sum_{i=1}^{M}s_{i}(l\Delta t-\tau_{k})\notag\\
&+\!\sqrt{1-\rho^2_{1}}\bm{\varepsilon}_{1}x_{u}(l\Delta t)+\bar{\bm{y}}_{self}(l\Delta t)\notag\\
&+\bar{\bm{y}}_{boun}(l\Delta t)+\bm{z}(l\Delta t).
\end{align}
$\hat{\bm{y}}'_{rad}$ is given by
\begin{align}
\!\!\!\!\!\!\hat{\bm{y}}'_{rad}=&\left[\alpha_{k}\bm{a}^{T}(\theta_{k})\exp(j\omega_{k}\Delta t)\sum_{i=1}^{M}s_{i}(\Delta t-\tau_{k}),...,\right.\notag\\
&\left.\alpha_{k}\bm{a}^{T}(\theta_{k})\exp(j\omega_{k}L\Delta t)\sum_{i=1}^{M}s_{i}(L\Delta t-\tau_{k})\right]^{T}.
\end{align}
$\bm{y}'_{rad}$ is given by
\begin{align}
\bm{y}'_{rad}=&\left[\exp(j\omega_{k}\Delta)\sum_{i=1}^{M}s_{i}(\Delta t-\tau_{k}),...,\right.\notag\\
&\left.\exp(j\omega_{k}L\Delta t)\sum_{i=1}^{M}s_{i}(L\Delta t-\tau_{k})\right]^{T}.
\end{align}
$\hat{\bm{z}}$ denotes the combined interference and noise. It is given by
\begin{align}
\hat{\bm{z}}=&\left[\sqrt{1-\rho^2_{1}}\bm{\varepsilon}^{T}_{1}x_{u}(\Delta t)+\bar{\bm{y}}^{T}_{self}(\Delta t)+\bar{\bm{y}}^{T}_{boun}(\Delta t)\right.\notag\\
&+\bm{z}^{T}(\Delta t),...,\sqrt{1-\rho^2_{1}}\bm{\varepsilon}^{T}_{1}x_{u}(L\Delta t)+\bar{\bm{y}}^{T}_{self}(L\Delta t)\notag\\
&+\bar{\bm{y}}^{T}_{boun}(L\Delta t)+\bm{z}^{T}(L\Delta t)\bigg]^{T}
\end{align}

Based on the sampled signals in (\ref{rad-sample}), to obtain the Cram$\acute{\textup{e}}$r-Rao bounds on direction, range and velocity estimation, we first need to obtain the Fisher information matrix (FIM), which is the inverse of the CRB matrix, on the $\theta_{k}$, $\tau_{k}$ and $\omega_{k}$. For the complex baseband signals, according to Eq. (15.52) in \cite{kay1993fundamental}, the FIM can be calculated by
\begin{equation}\label{FIM}
\bm{F}=2\cdot\text{Re}\{\bm{d}^{*}(\bm{\chi})(\bm{\Gamma}^{-1}\otimes\bm{\Lambda}^{-1})\bm{d}^{T}(\bm{\chi})\}
\end{equation}
where
\begin{align}
\bm{d}(\bm{\chi})=&\left[\bm{d}^{T}_{\alpha_{k}}(\bm{\chi}),\bm{d}^{T}_{\theta_{k}}(\bm{\chi}),\bm{d}^{T}_{\tau_{k}}(\bm{\chi}),\bm{d}^{T}_{\omega_{k}}(\bm{\chi})\right]^{T}\notag\\
=&\left[\frac{\partial{\hat{\bm{y}}'_{rad}}}{\partial{\alpha_{k}}^{T}},\frac{\partial{\hat{\bm{y}}'_{rad}}}{\partial{\theta_{k}}},\frac{\partial{\hat{\bm{y}}'_{rad}}}{\partial{\tau_{k}}},\frac{\partial{\hat{\bm{y}}'_{rad}}}{\partial{\omega_{k}}}\right]^T
\end{align}
$\bm{\chi}=[\alpha_{k}, \theta_{k},\tau_{k},\omega_{k}]^{T}$ denotes the vector of all the unknown parameters.
$\bm{\Gamma}$ and $\bm{\Lambda}$ denote the temporal and spatial domain covariance matrices of the combined noise and residual interference, respectively. They are given by
\begin{equation}
\bm{\Gamma} =
\left[
\begin{matrix}
\bar{\sigma}_{t_{1,1}}^{2} &\cdots & \bar{\sigma}_{t_{1,N}}^2 \\
\vdots &\ddots & \vdots \\
\bar{\sigma}^{2}_{t_{N,1}} &\cdots & \bar{\sigma}_{t_{N,N}} \\
\end{matrix}
\right],
\bm{\Lambda} =
\left[
\begin{matrix}
\bar{\sigma}_{s_{1,1}}^{2} &\cdots & \bar{\sigma}_{s_{1,N}}^2 \\
\vdots &\ddots & \vdots \\
\bar{\sigma}^{2}_{s_{N,1}} &\cdots & \bar{\sigma}^2_{s_{N,N}} \\
\end{matrix}
\right]
\end{equation}
$\alpha_{k}=\text{Re}\{\alpha_{k}\}+\text{Im}\{\alpha_{k}\}j$ compose of the real part and imaginary part. Then, we have

\begin{align}
\bm{d}_{\alpha_{k}}(\bm{\chi})=&\left[\frac{\partial{\hat{\bm{y}}'_{rad}}}{\partial{\text{Re}\{\alpha_{k}\}}},\frac{\partial{\hat{\bm{y}}'_{rad}}}{\partial{\text{Im}\{\alpha_{k}\}}}\right]^T\notag\\
=&\left[[1,j]\otimes\hat{\bm{y}}'_{rad}\otimes\bm{a}(\theta_{k})\right]^T
\end{align}
\begin{align}
\bm{d}_{\theta_{k}}(\bm{\chi})=\left[\alpha_{k}\cdot\hat{\bm{y}}'_{rad}\otimes\frac{\partial \bm{a}(\theta_{k})}{\partial\theta_{k}}\right]^{T}
\end{align}
\begin{align}
\bm{d}_{\tau_{k}}(\bm{\chi})=\left[\alpha_{k}\cdot \frac{\partial\hat{\bm{y}}'_{rad}}{\partial \tau_{k}}\otimes\bm{a}(\theta_{k})\right]^{T}
\end{align}
\begin{align}
\bm{d}_{\omega_{k}}(\bm{\chi})=\left[\alpha_{k}\cdot \frac{\partial\hat{\bm{y}}'_{rad}}{\partial \omega_{k}}\otimes\bm{a}(\theta_{k})\right]^{T}
\end{align}
where $\otimes$ denotes the Kronecker product. By further analysis, we can obtain the general radar estimation rate on the direction, range and velocity estimation of the target.

\begin{theorem}
The radar estimation rates of the direction $\theta_{k}$, range $d_{1}$ and velocity $v_{1}$ are given by
\begin{equation}\label{est-rate-theta}
R_{\theta_{k}}=\frac{1}{T_{c}}\log_{2}(1+\frac{\sigma_{\theta_{k}}^{2}}{\bm{C}(\theta_{k},\theta_{k})}),
\end{equation}

\begin{equation}\label{est-rate-tau}
R_{d_{k}}=\frac{1}{T_{c}}\log_{2}(1+\frac{4}{c_{0}^2}\cdot\frac{\sigma_{\tau_{k}}^{2}}{\bm{C}(\tau_{k},\tau_{k})}),
\end{equation}

\begin{equation}\label{est-rate-omega}
R_{v_{k}}=\frac{1}{T_{c}}\log_{2}(1+\frac{4\omega^2_{c}}{c_{0}^2}\cdot\frac{\sigma_{\omega_{k}}^{2}}{\bm{C}(\omega_{k},\omega_{k})}),
\end{equation}
respectively, where $c_{0}$ denotes the waveform propagation speed, $\omega_{c}$ denotes the carrier angular frequency,
\begin{equation}
\bm{C}(\theta_{k},\theta_{k})=\frac{\left[\frac{\partial\bm{a}^{\dag}(\theta_{k})}{\partial\theta_{k}}\Pi_{\bm{\Lambda}}^{\bot}\frac{\partial\bm{a}(\theta_{k})}{\partial\theta_{k}}\right]^{-1}}{2|\alpha_{k}|^2\bm{y}^{\dag}_{rad}\bm{\Gamma}^{-1}\bm{y}_{rad}},
\end{equation}
where
\begin{equation}
\Pi_{\bm{\Lambda}}^{\bot}=\bm{\Lambda}^{-1}-\frac{\bm{\Lambda}^{-1}\bm{a}(\theta_{k})\bm{a}^{\dag}(\theta_{k})\bm{\Lambda}^{-1}}{\bm{a}^{\dag}(\theta_{k})\bm{\Lambda}^{-1}\bm{a}(\theta_{k})}.
\end{equation}
$\bm{C}(\tau_{k},\tau_{k})$ and $\bm{C}(\omega_{k},\omega_{k})$ are given by

\begin{align}\label{reduce-crb}
CRB(\tau_{k},\omega_{k})=&
\begin{bmatrix}
\bm{C}(\tau_{k},\tau_{k})&\bm{C}(\tau_{k},\omega_{k})\\
\bm{C}(\tau_{k},\omega_{k})^{T}&\bm{C}(\omega_{k},\omega_{k})
\end{bmatrix}\notag\\
=&
\begin{bmatrix}
\bm{F}'(\tau_{k},\tau_{k})&\bm{F}'(\tau_{k},\omega_{k})\\
\bm{F}'(\tau_{k},\omega_{k})^{T}&\bm{F}'(\omega_{k},\omega_{k})
\end{bmatrix}
^{-1}
\end{align}
where $\text{CRB}_{\tau_{k},\omega_{k}}$ denotes the reduced CRB on $\tau_{k}$ and $\omega_{k}$. $\bm{F}'$ denotes the reduced FIM on $\tau_{k}$ and $\omega_{k}$ and
\begin{equation}
\bm{F}'(\tau_{k},\tau_{k})=\left[2|\alpha_{k}|^2\bm{a}^{\dag}(\theta_{k})\bm{\Lambda}^{-1}\bm{a}(\theta_{k})\right]\cdot\left(\frac{\partial\bm{y}^{\dag}_{rad}}{\partial\tau_{k}}\Pi^{\bot}_{\bm{\Gamma}}\frac{\partial\bm{y}_{rad}}{\partial\tau_{k}}\right),
\end{equation}
\begin{equation}
\bm{F}'(\tau_{k},\omega_{k})\!=\!\left[2|\alpha_{k}|^2\bm{a}^{\dag}(\theta_{k})\!\bm{\Lambda}^{-1}\!\bm{a}(\theta_{k})\right]\!\text{Re}\!\left\{\!\frac{\partial\bm{y}^{\dag}_{rad}}{\partial\tau_{k}}\Pi^{\bot}_{\Gamma}\frac{\partial\bm{y}_{rad}}{\partial\omega_{k}}\!\right\}\!,
\end{equation}
\begin{equation}
\bm{F}'(\omega_{k},\omega_{k})=\left[2|\alpha_{k}|^2\bm{a}^{\dag}(\theta_{k})\bm{\Lambda}^{-1}\bm{a}(\theta_{k})\right]\cdot\left(\!\frac{\partial\bm{y}^{\dag}_{rad}}{\partial\omega_{k}}\Pi^{\bot}_{\bm{\Gamma}}\frac{\partial\bm{y}_{rad}}{\partial\omega_{k}}\!\right),
\end{equation}
where
\begin{equation}
\Pi_{\bm{\Gamma}}^{\bot}=\bm{\Gamma}^{-1}-\frac{\bm{\Gamma}^{-1}\bm{y}_{rad}\bm{y}^{\dag}_{rad}\bm{\Gamma}^{-1}}{\bm{y}^{\dag}_{rad}\bm{\Gamma}^{-1}\bm{y}_{rad}}.
\end{equation}
\end{theorem}
\begin{proof}
See Appendix A.
\end{proof}

Here, it is noted that off-diagonal entries of the $\text{CRB}_{\tau_{k},\omega_{k}}$ do not equal to zero. The reason is that during the estimation procedures, the estimation of $\tau_{k}$ and $\omega_{k}$ cannot be decoupled, i.e., the value of estimated $\tau_{k}$ affects the value of estimated $\omega_{k}$ \cite{wittman2012fisher}. Hence, the estimation rate of the range $d_{1}$ and the velocity $v_{1}$ appear in pair-wise. However, the direction estimation is independent of the range and velocity estimations, i.e., the estimated $\theta_{k}$ will not change no matter $\tau_{k}$ and $\omega_{k}$ are known or not.

Next, we consider the following special case.

\emph{Special case 1:}
\begin{enumerate}
\item The combined noise and residual interference $\hat{z}$ is complex Gaussian distributed with covariance $Cov[\hat{z}]=\left[(1-\rho^{2}_{1})\Omega_{ul}P_{ul}+K_{self}P_{dl}+\sigma^{2}_{z}\right]\bm{I}_{L\times N}$, i.e.,
    \begin{equation}
\bm{\Gamma} =
\left[
\begin{matrix}
\bar{\sigma}_{t_{0,0}}^{2} &\cdots & 0 \\
\vdots &\ddots & \vdots \\
0 &\cdots & \bar{\sigma}^{2}_{t_{0,0}} \\
\end{matrix}
\right],
\bm{\Lambda} =
\left[
\begin{matrix}
\bar{\sigma}_{s_{0,0}}^{2} &\cdots & 0 \\
\vdots &\ddots & \vdots \\
0 &\cdots & \bar{\sigma}^2_{s_{0,0}} \\
\end{matrix}
\right],
\end{equation}
$\bar{\sigma}_{t_{0,0}}^{2}\times \bar{\sigma}_{s_{0,0}}^{2} = \left[(1-\rho^{2}_{1})\Omega_{ul}P_{ul}+K_{self}P_{dl}+\sigma^{2}_{z}\right]$.
\item  The receiver antenna array is uniformly placed with half waveform length, i.e., $\bm{a}(\theta_{k})=\left[1,e^{j\pi\sin(\theta_{k})},...,e^{j(N-1)\pi\sin(\theta_{k})}\right]^{T}$
\end{enumerate}

\begin{theorem}
The estimation rate of the direction $\theta_{k}$, range $d_{1}$ and velocity $v_{1}$ under the \emph{special case 1} is given by (\ref{est-rate-theta}), (\ref{est-rate-tau}) and (\ref{est-rate-omega}), respectively, where

\begin{equation}
\bm{C}(\theta_{k},\theta_{k})=\frac{6\left[(1-\rho^2_{2})\Omega_{ul}P_{ul}+(K_{self}\!+\!K_{boun})P_{dl}\!+\!\sigma^{2}_{z}\right]}{|\alpha_{k}|^2N(N^2-1)f_{B}T_{c}P_{rad}\pi^{2}\cos^2(\theta_{k})},
\end{equation}
%\newpage
and $\bm{C}(\tau_{k},\tau_{k})$ and $\bm{C}(\omega_{k},\omega_{k})$ are given by (\ref{reduce-crb}), where
\begin{align}
\!\!\!\!\!\bm{F}'(\tau_{k},\tau_{k})=&\frac{2N|\alpha_{k}|^2}{(1-\rho_{2}^2)\Omega_{ul}P_{ul}+(K_{self}\!+\!K_{boun})P_{dl}+\sigma_{z}^2}\notag\\
&\times\left[\left(\frac{1}{\Delta t}\sum_{i=1}^{M}\int_{0}^{T_{c}}\left|\frac{\partial s_{i}(t)}{\partial t}\right|^2dt\right)\right.\notag\\
&\left.-\frac{1}{\xi_{r}}\left|\frac{1}{\Delta t}\sum_{i=1}^{M}\int_{0}^{T_{c}}\frac{\partial s^{\dag}_{i}(t)}{\partial t}s_{i}(t)dt\right|^2\right],
\end{align}

\begin{align}
\!\!\!\!\!\bm{F}'(\tau_{k},\omega_{k})=&\frac{2N|\alpha_{k}|^2}{(1-\rho_{2}^2)\Omega_{ul}P_{ul}+(K_{self}+K_{boun})P_{dl}+\sigma_{z}^2}\notag\\
&\times\text{Im}\left\{\frac{1}{\Delta t}\sum_{i=1}^{M}\int_{0}^{T_{c}}t\frac{\partial s^{\dag}_{i}}{\partial t}s_{i}(t)dt\right.\notag\\
&-\frac{1}{\xi_{r}}\left(\frac{1}{\Delta t}\sum_{i=1}^{M}\int_{0}^{T_{c}}\frac{\partial s^{\dag}_{i}(t)}{\partial t}s_{i}(t)dt\right)\notag\\
&\left.\times\left(\frac{1}{\Delta t}\sum_{i=1}^{M}\int_{0}^{T_{c}}t|s_{i}(t)|^2dt\right)\right\},
\end{align}

\begin{align}
\!\!\!\!\!\bm{F}'({\omega_{k},\omega_{k}})=&\frac{2N|\alpha_{k}|^2}{(1-\rho_{2}^2)\Omega_{ul}P_{ul}+(K_{self}+K_{boun})P_{dl}+\sigma_{z}^2}\notag\\
&\times\left[\frac{1}{\Delta t}\sum_{i=1}^{M}\int_{0}^{T_{c}}t^2|s_{i}(t)|^2dt\right.\notag\\
&\left.-\frac{1}{\xi_{r}}\left(\frac{1}{\Delta t}\sum_{i=1}^{M}\int_{0}^{T_{c}}t|s_{i}(t)|^2dt\right)^2\right],
\end{align}
where $\xi_{r}=f_{B}T_{c}\delta P_{rad}$.
\end{theorem}
\begin{proof}
Substituting $\bm{\Lambda}$, $\bm{\Gamma}$ into (\ref{reduce-crb}) and extending $\hat{\bm{y}}'_{rad}$, and after simplification, \emph{Theorem 2} can be easily obtained.
\end{proof}

Here, we get the exact closed-form estimation rate on the direction estimation of $\theta_{k}$. However, no closed-form expressions are obtained for the time delay $\tau_{k}$ and Doppler shift $\omega_{k}$. It is noted that the estimation rate $R_{\theta_{k}}$ of the direction $\theta_{k}$ is related to the power of the received radar return signals and unrelated to the specific signal forms. However, the estimation rates $R_{\tau_{k}}$ and $R_{\omega_{k}}$ of $\tau_{k}$ and $\omega_{k}$ are determined by the specific signal forms, implying that the obtained estimation rates can serve as metric for the radar waveform design. To gain some more insights on the joint radar and communication system, we consider the specific case 2 based on the specific case 1.

\emph{Special case 2:} The radar operation adopts the code division orthogonal linear frequency modulated (LFM) signals with duty factor $\delta$ \cite{kocjancic2018,rao2013}.  The LFM radar signals are denoted as
\begin{equation}
s_{i}(t)=\sqrt{\frac{P_{rad}}{M}}\sum_{k=0}^{K-1}s'_{i}(t-kT_{R})
\end{equation}
where $T_{R}$ is the pulse repetition interval. $K$ denotes the number of the radar pulses with each CPI. $s'_{i}(t)$ denotes the radar waveform emitted by the \emph{i}-th antenna, it is given by
\begin{equation}\label{74664632}
s'_{i}(t)=\exp\left[j2\pi\frac{f_{B}}{T_{0}}(t-\frac{1}{2}T_{0})^2+j\frac{\pi}{2}b_{i}(t)\right]\cdot \emph{Rect}_{T_{0}}(t)
\end{equation}
where $b_{i}(t), i\in\{1,...,M\}$ denotes the code sequence to achieve orthogonality between different radar signal waveforms \cite{rao2013}. $\emph{Rect}_{T_{0}}(t)$ denotes the rectangular window function between $[0, T_{0}]$, i.e.,
\begin{equation}
\emph{Rect}_{T_{0}}(t)=\left\{
\begin{array}{l}
\vspace{2mm}
1,\qquad t\in[0, T_{0}]\\
0,\qquad else
\end{array}
\right.
\end{equation}

With the above radar signals, we can obtain the exact closed-form estimation rate.

\begin{theorem}
The exact closed-form estimation rates of the range $d_{1}$ and velocity $v_{1}$ under the \emph{special case 1} and \emph{special case 2} are given by (\ref{est-rate-tau}) and (\ref{est-rate-omega}), respectively, where
\begin{equation}
\text{C}(\tau_{k},\tau_{k})\!=\!\frac{\bar{\sigma}_{z}^2}{2N|\alpha_{k}|^2\delta T_{C}P_{rad}f_{B}^3}\cdot\frac{3[(K^2-1)T^2_{R}+T_{0}^2]}{\pi^2[(K^2\!-\!1)T^{2}_{R}\!-\!3T_{0}^2]},
\end{equation}
\begin{equation}
\text{C}(\omega_{k},\omega_{k})=\frac{\bar{\sigma}_{z}^2}{2N|\alpha_{k}|^2\delta T_{C}P_{rad}f_{B}}\cdot\frac{12}{(K^2-1)T^{2}_{R}-3T_{0}^2},
\end{equation}

\begin{proof}
{\color{black}Please refer to Appendix B.}
\end{proof}
\end{theorem}

\section{Numerical Results}
In the numerical simulations, for the radar operation, we focus on the line-of-sight channel. However, for the communication operation, we consider the non-line-of-sight Rayleigh fading environment. Taking the large scale fading and shadowing effects into account, we consider the following specific channel model \cite{8486331,tsai2011path},
\begin{equation}
h_{k}=\frac{10^{\beta_{k}/10}}{1+(d_{k}/d_0)^l}\cdot\bar{h}_{k}, \quad k \in\{d,u\}
\end{equation}
where $10^{\beta_{k}/10}$ characterize the shadowing effects in log-normal distribution with standard deviation $\sigma_{s}$ dB, i.e., $\beta_{w}\sim \mathcal{N}(0,\sigma_{s}^2)$ which is normal distribution in dB. $d_{k}$ denotes the distance of the corresponding link. $d_0$ denotes the reference distance. $l$ is the path loss exponent. $\bar{h}_{k}$ characterizes the small fading, which is standard unit complex distributed, i.e., $\bar{h}_{k} \sim \mathcal{CN}(0, 1)$. {\color{black}Hence, we have the variance of the estimated channel as}
\begin{align}
\Omega_{dl}=&\mathbb{E}\left\{\frac{10^{\beta_{d}/10}}{1+(d_{d}/d_0)^l}\right\},\\ \Omega_{ul}=&\mathbb{E}\left\{\frac{10^{\beta_{u}/10}}{1+(d_{u}/d_0)^l}\right\}, 
\end{align}

In the numerical simulation, the parameters are defaulted as follows (unless specified otherwise). For the large-scale fading parameters, we set $\beta_{k}=8$ dB, $d_{0}=200$ m, $l=3.8$, and $d_{d}=d_{u}=500$ m. For the channel estimation, we set $\rho_{k}=0.95$. For the system parameters, the numbers of antennas are set as $M=32$, $N=64$, the maximum transmission powers are set as $P_{rad}=20$ dBm, $P_{dl}=23$ dBm, $P_{ul}= 23$ dBm, assuming the interference can be suppressed to the noise level, i.e., $K_{self}=-23$ dB, $K_{rad}=-20$ dB, and $K_{boun}=-23$ dB, $K_{co}=-23$ dB, the variance of the noise $\sigma_{0}^2=1$, the shared bandwidth is set as $f_{B}=5$ MHz, the radar pulse relevant parameters are set as $T_{C}=50$ ms, $T_{R}=1$ ms, $\delta=0.1$, $T_{0}=100$ us referring to \cite{kumari2018ieee802,chiriyath2016inner,dogandzic2001}.

{\color{black}To investigate the feasibility of full-duplex joint radar and communication system under the impact of imperfect interference cancellation, the performance regions of the joint system in terms of DL and UL communication rate and radar estimation rate are needed. We obtain the bound of the achievable region by setting the transmission power of one of the links to maximum (unless noted otherwise, it is set to values above), and then increasing the power of the other link. It is noted that the system performance is always within the bounded region as long as either of the transmission powers is not maximum.}

\begin{figure}[ht]
  \centering
  \includegraphics[width=0.6\linewidth]{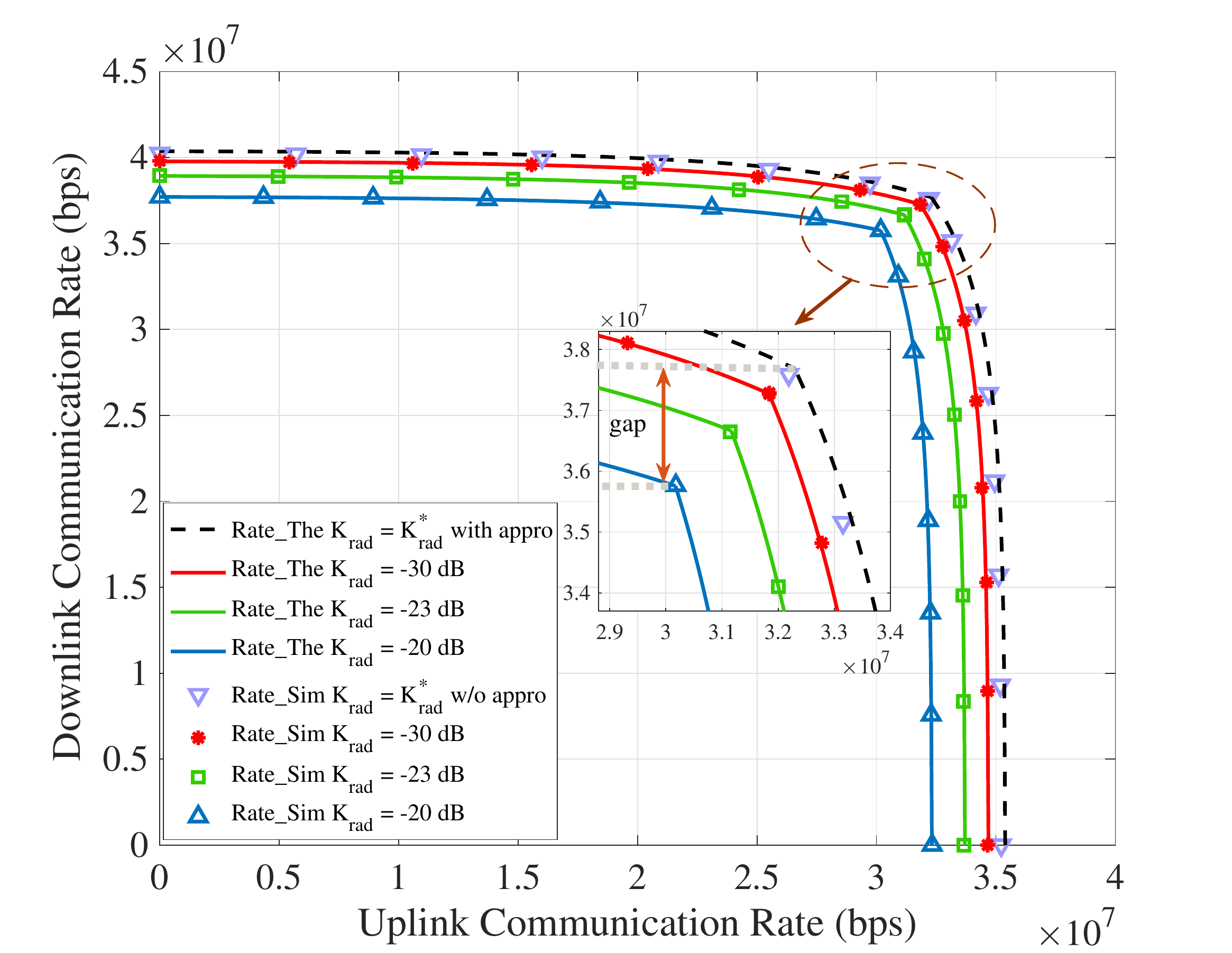}
   \caption{\color{black}Downlink and uplink rate region under the radar interference.}\label{figure1}
\end{figure}
\textcolor{black}{In Fig. \ref{figure1}, we demonstrate the downlink and uplink communication rate regions under the effects of the different radar interference cases, where the exact value of radar interference suppression level "$K_{rad}=K_{rad}^{*}$ with appro" is calculated by the closed-form theoretical limit in (\ref{radar_suppresion_level}) and "$K_{rad}=K_{rad}^{*}$ w/o appro" is simulated without approximation made in (\ref{appro_level}). By setting the transmission power of one communication link to be 20 dBm and increasing power of the other from 0 to 20 dBm, we can obtain the bounds of the communication region.}  It is found that the level of the radar interference limits the communication rate regions. However, when the radar signals can be suppressed to the noise level, e.g., $K_{rad}=-30$ dB, the communication operation almost achieves the same rate regions compared to that of achievable perfect radar cancellation reaching Cram$\acute{\textup{e}}$r-Rao bound, i.e., the case $K_{rad}=K_{rad}^{*}$. \textcolor{black}{ Note that the achievable value of radar interference suppression level $K_{rad}$ is always larger than $K_{rad}^{*}$. As can be seen, the bound of simulated suppression level without approximation, i.e. $K_{rad}=K^{*}_{rad}$ w/o appro, is always lower than the theoretical one which result from the overestimation of interference suppression capability when the approximation of residual signal is adopted. Moreover, the discrete Monte Carlo simulations match with theoretical evaluations, therefore the tightness of Jensen's inequality is demonstrated.} This result reveals that sharing the spectrum of the radar system to the communication system introduce significant communication rate regions though the radar signals act as interferences.

\begin{figure}[ht]
  \centering
  \includegraphics[width=0.6\linewidth]{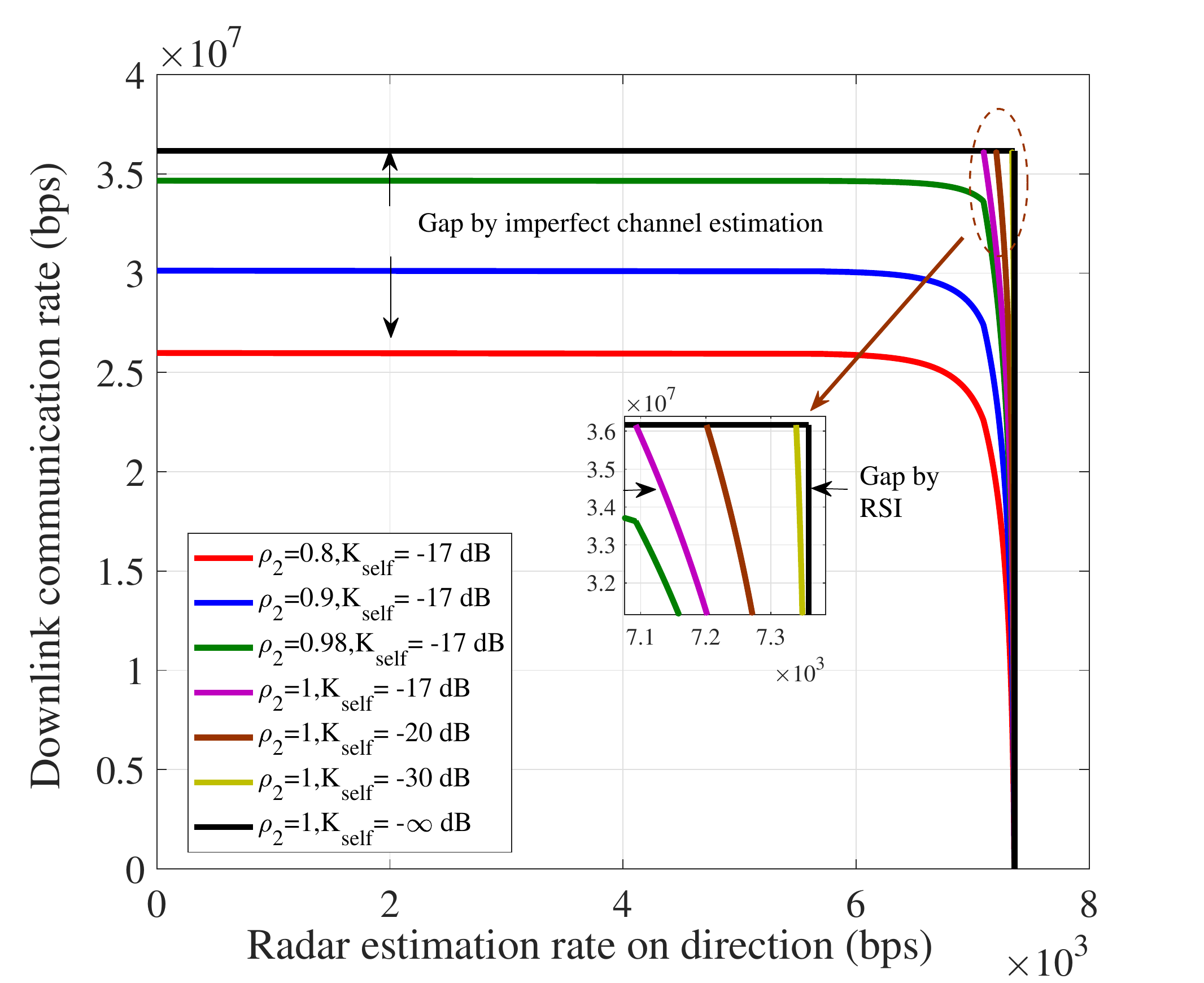}
  \caption{Downlink communication rate and direction estimation rate region.}\label{figure2}
\end{figure}

\textcolor{black}{Fig. \ref{figure2} presents the joint downlink communication and radar estimation rate regions with different channel estimation accuracies and different RSI regimes.} This figure reveals that the joint rate regions are limited by the accuracy of the channel estimation as well as the self-interference cancellation capability. However, it is observed that the RSI only has slight impacts on the joint rate regions, for example, when self-interference is cancelled by -17 dB, the rate region is almost the same compared to that with perfect self-interference cancellation.

\begin{figure}[ht]
  \centering
  \includegraphics[width=0.6\linewidth]{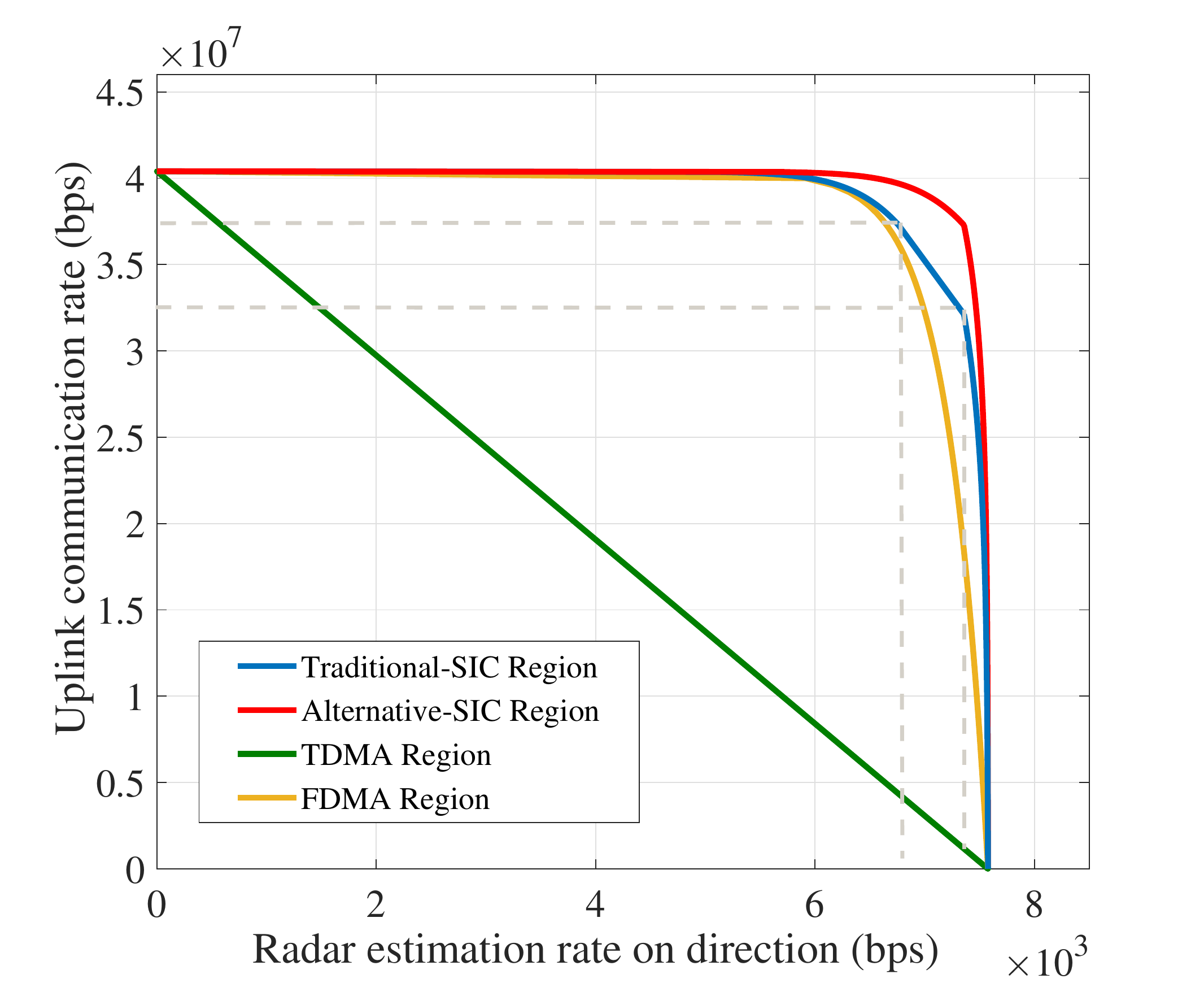}
  \caption{Uplink communication rate and direction estimation rate region.}\label{figure3}
\end{figure}

\begin{figure}[ht]
  \centering
  \includegraphics[width=0.6\linewidth]{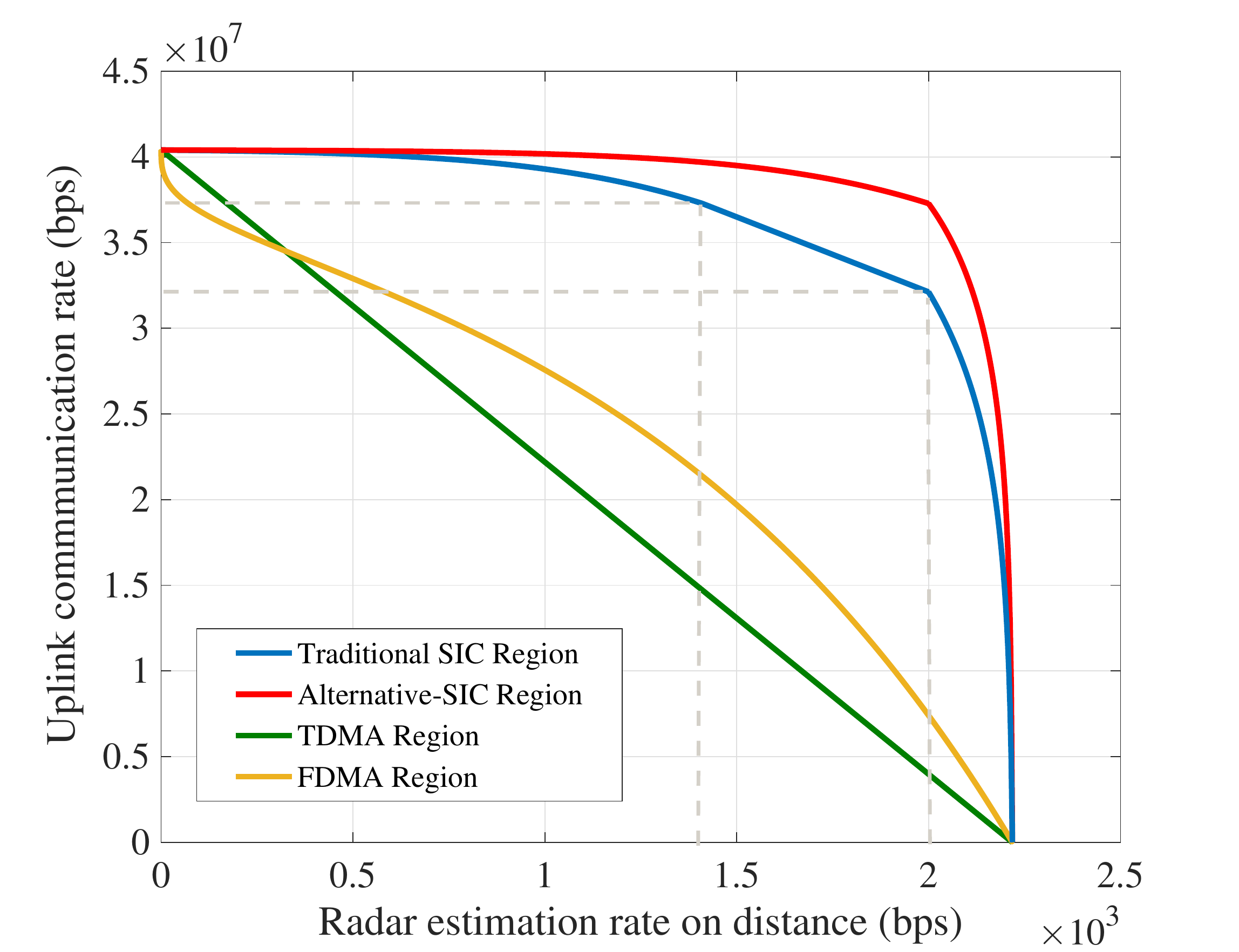}
  \caption{Uplink communication rate and distance estimation rate region.}\label{figure4}
\end{figure}

\textcolor{black}{In fig. \ref{figure3}, we compare the performance of the proposed alternative-SIC scheme in terms of joint uplink and estimation rate regions with other benchmark methods, i.e. traditional SIC scheme, time division multiple access (TDMA) and frequency division multiple access (FDMA) schemes.} With the traditional SIC scheme, the ComRad first decodes communication signals, then subtracts the communication signals and estimates the radar targets or in reverse order. With the alternative-SIC scheme, the ComRad first suppress the radar interference with known waveform information, then decode the communication signals and subtract the communication waveform, finally estimates the targets' parameters. For the TDMA scheme, the duration is divided into two parts, one part is only for communication while the other is only for the radar operation. For the FDMA scheme, the whole system bandwidth is divided into two sub-bands, one sub-band is only for communication while the other is only for radar operation. It is noted that compared to the orthogonal schemes, i.e., TDMA and FDMA, the joint schemes achieve larger rate regions, which reveals the potential on the rate regions of the spectrum sharing between the radar system and communication system. In addition, compared to the traditional SIC scheme, the proposed alternative-SIC scheme could further enlarge the rate regions with the gains depending on the interference cancellation capability.

\begin{figure}[ht]
  \centering
  \includegraphics[width=0.6\linewidth]{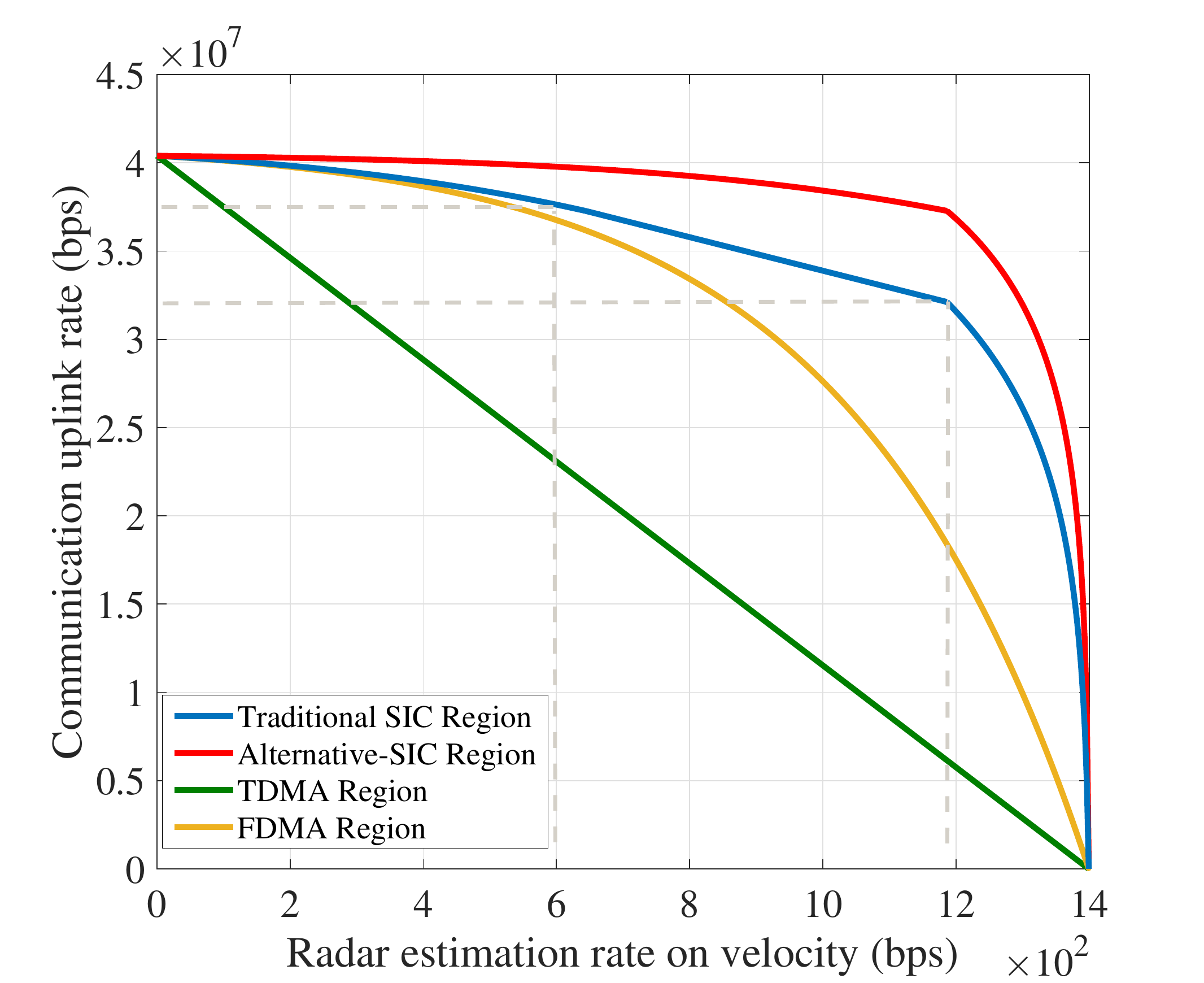}
  \caption{Uplink communication rate and velocity estimation rate region.}\label{figure5}
\end{figure}

Fig. \ref{figure4} and Fig. \ref{figure5} illustrate the joint rate regions in terms of the uplink rate and radar estimation rate on range and velocity, respectively, under the proposed alternative-SIC, traditional SIC, TDMA and FDMA schemes. It is noted that for the orthogonal working manner, the FDMA manner is not guaranteed to achieve larger rate region compared to the TDMA scheme. In addition, we note the advantages of the joint working manner compared to the isolating scheme in terms of the joint rate regions. \textcolor{black}{Both figures demonstrates the tradeoff between the uplink communication rate and radar estimation rate, i.e., the communication operation can achieve significant rate performance at the slight cost of the radar estimation performance while working jointly.} The results indicate that under certain radar estimation constraint, the spectrum sharing between radar and communication is feasible and gainful.
\begin{figure}[ht]
  \centering
  \includegraphics[width=0.6\linewidth]{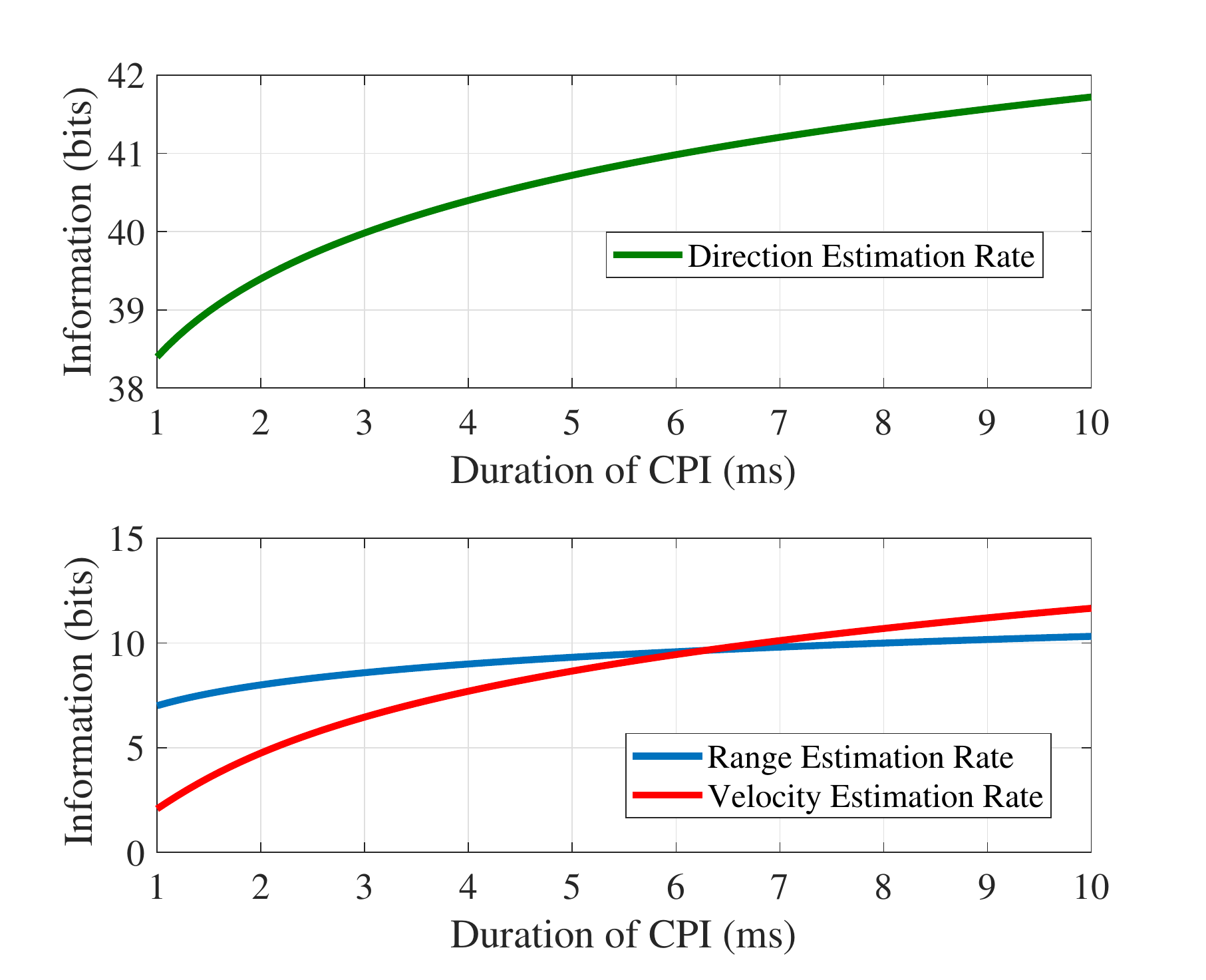}
  \caption{Estimation rates versus the duration of CPI.}\label{figure6}
\end{figure}

\begin{figure}[ht]
  \centering
  \includegraphics[width=0.6\linewidth]{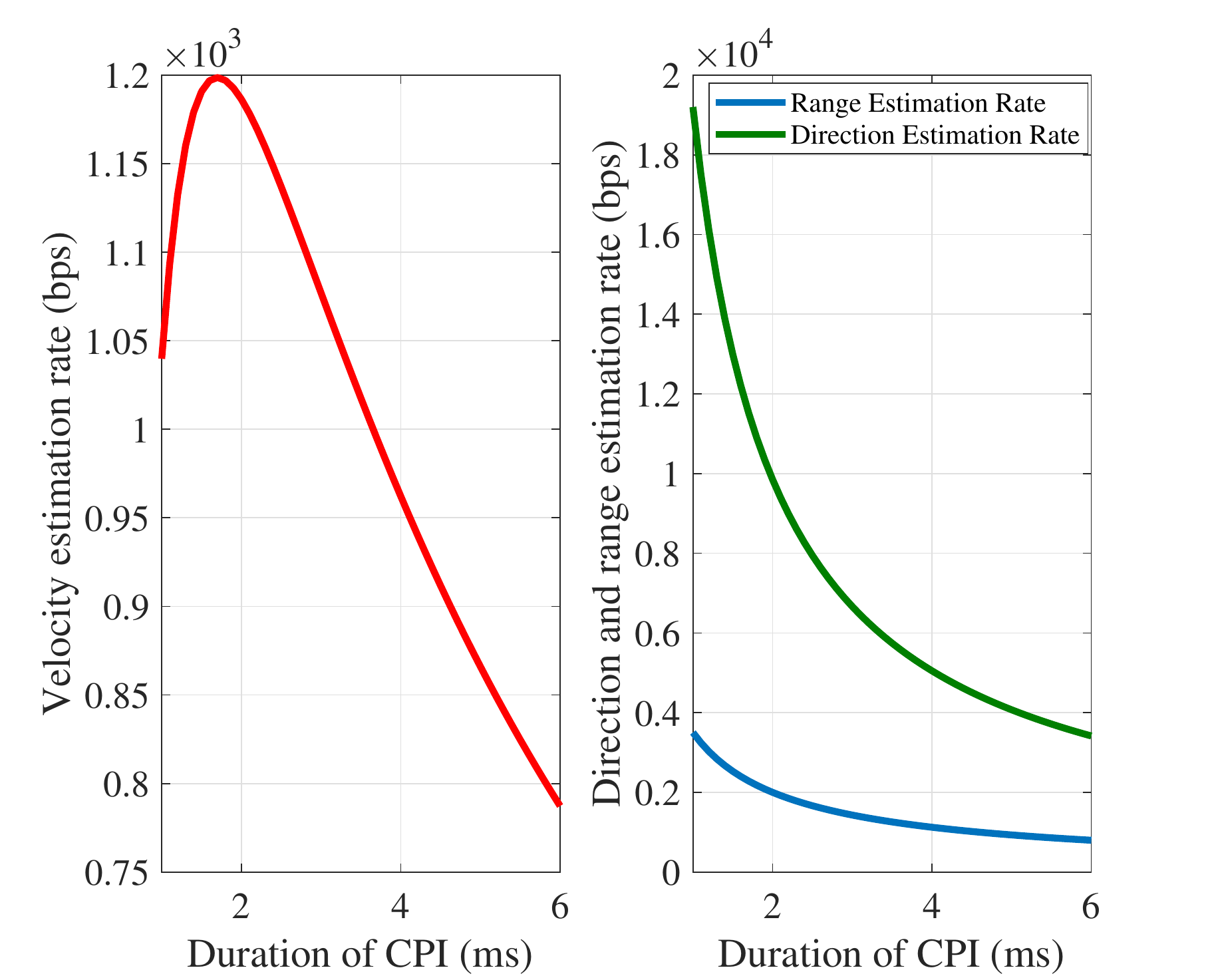}
  \caption{Estimation rates versus the duration of CPI.}\label{figure7}
\end{figure}

In Fig. \ref{figure6} and Fig. \ref{figure7}, we show the variation of the amount of the information about the target obtained within each CPI and within unit time, i.e., 1 second, with regarding to the duration of CPI, respectively. Different from the communication system where the duration of time slots do not affect the estimation rate, the estimation accuracy in radar system is related to the duration of CPI. Fig. \ref{figure6} demonstrates that the amount of information obtained about the target increases with the duration of CPI at different speed, which is consistent with our intuition that the more snapshots of the radar signals used for estimation, the higher accuracy can be achieved. However, as shown in Fig. \ref{figure7}, the increase of the duration of CPI does not guarantee the increase of the estimation rate within unit time. Interestingly, for the velocity, the estimation rate firstly increases, after reaching the maximum value and finally decreases with the duration of the CPI. For the range and direction, the estimation rates always decrease with the duration of each CPI.

We interpret this as the fundamental difference between velocity estimation and range and direction estimation. Unlike range and direction, which are estimated by instantaneous received signal, velocity is estimated by the observation of signal frequency offset in a small time interval which results from doppler effect. 
This small time interval is helpful to resolve the doppler effect which transfers more velocity information. Thus, a small rise in duration of CPI increases the velocity estimation rate, but a long CPI incurs inefficiency.

\section{Conclusion}
In this paper, the joint FD communication and radar multi-antenna system was studied. To facilitate the dual functions of the ComRad, a receiver structure based on the alternative-SIC scheme was proposed  by viewing the uplink channel and radar return channel as MAC. Based on the proposed structure, the uplink and downlink communication rate was derived under the MRT/MRC beamforming scheme considering self-interference and imperfect channel estimation. The achievable suppression level of the radar return signal as the interference to uplink communication was derived. For the radar operations, the estimation information rates in terms of the distance, direction and velocity estimations were obtained for both the general case and the special case with  uniform linear antenna array and LFM chirp radar signals. By the joint rate regions, the interaction between the radar operation and the communication operation were demonstrated. The numerical simulation results showed that sharing the radar frequency bands to the communication operation bring significant transmission capacity while the radar performance can be guaranteed even if communication signals act as interference.

\appendices
\section{Proof Of Theorem 1}
To obtain the CRB on $\theta_{k}$, $\tau_{k}$ and $\omega_{k}$, we have to obtain the FIM first, which is given by
\small{
\begin{equation}
\text{FIM}=
\begin{bmatrix}
\vspace{1mm}
  \bm{F}(\alpha_{k},\alpha_{k})&\bm{F}(\alpha_{k},\theta_{k})&\bm{F}(\alpha_{k},\tau_{k})&\bm{F}(\alpha_{k},\omega_{k}) \\
  \vspace{1mm}
  \bm{F}(\alpha_{k},\theta_{k})^{T}&\bm{F}(\theta_{k},\theta_{k}) &\bm{F}(\theta_{k},\tau_{k})&\bm{F}(\theta_{k},\omega_{k}) \\
  \vspace{1mm}
 \bm{F}(\alpha_{k},\tau_{k})^{T}&\bm{F}(\theta_{k},\tau_{k})^T &\bm{F}(\tau_{k},\tau_{k}) & \bm{F}(\tau_{k},\omega_{k})\\
  \bm{F}(\alpha_{k},\omega_{k})^T&\bm{F}(\theta_{k},\omega_{k})^T&\bm{F}(\tau_{k},\omega_{k})^T&\bm{F}(\omega_{k},\omega_{k})
\end{bmatrix}
\end{equation}}

where
\begin{equation}
\bm{F}(\alpha_{k},\alpha_{k})=2\cdot\hat{\bm{y}}_{rad}^{\dag}\bm{\Gamma}^{-1}\hat{\bm{y}}_{rad}\cdot\bm{a}(\theta_{k})^{\dag}\bm{\Lambda}^{-1}\bm{a}(\theta_{k})\cdot\bm{I}_{2}
\end{equation}
\begin{multline}
\bm{F}(\alpha_{k},\theta_{k})=2\cdot\text{Re}\left\{\left[[1,-j]\alpha_{k}\bm{y}_{rad}^{\dag}\bm{\Gamma}^{-1}\bm{y}_{rad}\right]^{T}\right.\\
\left.\times\bm{a}^{\dag}(\theta_{k})\bm{\Lambda}^{-1}\frac{\partial\bm{a}(\theta_{k})}{\partial\theta_{k}}\right\}
\end{multline}
\begin{multline}
\bm{F}(\alpha_{k},\tau_{k})=2\cdot\text{Re}\left\{\left[[1,-j]\alpha_{k}\bm{y}_{rad}^{\dag}\bm{\Gamma}^{-1}\frac{\partial\bm{y}_{rad}}{\partial\tau_{k}}\right]^{T}\right.\\
\left.\times\bm{a}^{\dag}(\theta_{k})\bm{\Lambda}^{-1}\bm{a}(\theta_{k})\right\}
\end{multline}
\begin{multline}
\bm{F}(\alpha_{k},\omega_{k})=2\cdot\text{Re}\left\{\left[ [1,-j]\alpha_{k}\bm{y}_{rad}^{\dag}\bm{\Gamma}^{-1}\frac{\partial\bm{y}_{rad}}{\partial\omega_{k}}\right]^{T}\right.\\
\left.\times\bm{a}^{\dag}(\theta_{k})\bm{\Lambda}^{-1}\bm{a}(\theta_{k})\right\}
\end{multline}
\begin{equation}
\bm{F}(\theta_{k},\theta_{k})=2\cdot|\alpha_{k}|^2\bm{y}^{\dag}_{rad}\bm{\Gamma}^{-1}\bm{y}_{rad}\frac{\partial\bm{a}^{\dag}(\theta_{k})}{\partial\theta_{k}}\bm{\Lambda}^{-1}\frac{\partial\bm{a}(\theta_{k})}{\partial\theta_{k}}
\end{equation}
\begin{equation}
\bm{F}(\theta_{k},\tau_{k})=2\cdot\text{Re}\left\{|\alpha_{k}|^2\bm{y}^{\dag}_{rad}\bm{\Gamma}^{-1}\frac{\partial\bm{y}_{rad}}{\partial\tau_{k}}\cdot\frac{\partial\bm{a}^{\dag}(\theta_{k})}{\partial\theta_{k}}\bm{\Lambda}^{-1}\bm{a}(\theta_{k})\right\}
\end{equation}
\begin{equation}
\bm{F}(\theta_{k},\omega_{k})=2\cdot\text{Re}\left\{|\alpha_{k}|^2\bm{y}^{\dag}_{rad}\bm{\Gamma}^{-1}\frac{\partial\bm{y}_{rad}}{\partial\omega_{k}}\cdot\frac{\partial\bm{a}^{\dag}(\theta_{k})}{\partial\theta_{k}}\bm{\Lambda}^{-1}\bm{a}(\theta_{k})\right\}
\end{equation}
\begin{equation}
\bm{F}(\tau_{k},\tau_{k})=2\cdot|\alpha_{k}|^2\frac{\partial\bm{y}^{\dag}_{rad}}{\partial\tau_{k}}\bm{\Gamma}^{-1}\frac{\partial\bm{y}_{rad}}{\partial\tau_{k}}\bm{a}^{\dag}(\theta_{k})\bm{\Lambda}^{-1}\bm{a}(\theta_{k})
\end{equation}
\begin{equation}
\bm{F}(\tau_{k},\omega_{k})=2\cdot\text{Re}\left\{|\alpha_{k}|^2\frac{\partial\bm{y}^{\dag}_{rad}}{\partial\tau_{k}}\bm{\Gamma}^{-1}\frac{\partial\bm{y}_{rad}}{\partial\omega_{k}}\bm{a}^{\dag}(\theta_{k})\bm{\Lambda}^{-1}\bm{a}(\theta_{k})\right\}
\end{equation}
\begin{equation}
\bm{F}(\omega_{k},\omega_{k})=2\cdot|\alpha_{k}|^2\frac{\partial\bm{y}^{\dag}_{rad}}{\partial\omega_{k}}\bm{\Gamma}^{-1}\frac{\partial\bm{y}_{rad}}{\partial\omega_{k}}\bm{a}^{\dag}(\theta_{k})\bm{\Lambda}^{-1}\bm{a}(\theta_{k})
\end{equation}

In this paper, we focus on the parameter estimation in terms of the range, direction and velocity of the target. Hence, the CRB matrix of $\theta_{k}$, $\bm{\phi}_{1}$ and $\bm{\phi}_{2}$ is given by

\begin{align}
\text{CRB}=&
\begin{bmatrix}
\vspace{2mm}
\bm{C}(\theta_{k},\theta_{k})&\bm{C}(\theta_{k},\tau_{k})&\bm{C}(\theta_{k},\omega_{k})\\
\vspace{2mm}
\bm{C}(\theta_{k},\tau_{k})^{T}&\bm{C}(\tau_{k},\tau_{k})&\bm{C}(\tau_{k},\omega_{k})\\
\bm{C}(\theta_{k},\omega_{k})^{T}&\bm{C}(\tau_{k},\omega_{k})^{T}&\bm{C}(\omega_{k},\omega_{k})
\end{bmatrix}\notag\\
=&\left\{
\begin{bmatrix}
\vspace{2mm}
  \bm{F}(\theta_{k},\theta_{k}) &\bm{F}(\theta_{k},\tau_{k})&\bm{F}(\theta_{k},\omega_{k}) \\
  \vspace{2mm}
 \bm{F}(\theta_{k},\tau_{k})^T & \bm{F}(\tau_{k},\tau_{k}) & \bm{F}(\tau_{k},\omega_{k}) \\
  \bm{F}(\theta_{k},\omega_{k})^T &\bm{F}(\tau_{k},\omega_{k}) ^T& \bm{F}(\omega_{k},\omega_{k})
\end{bmatrix}\right.\notag\\
&\left.-
\begin{bmatrix}
\vspace{2mm}
\bm{F}(\alpha_{k},\theta_{k})^{T}\\
\vspace{2mm}
\bm{F}(\alpha_{k},\tau_{k})^{T}\\
\bm{F}(\alpha_{k},\omega_{k})^{T}
\end{bmatrix}
 \bm{F}(\alpha_{k},\alpha_{k})^{-1}
\begin{bmatrix}
\vspace{2mm}
\bm{F}(\alpha_{k},\theta_{k})^{T}\\
\vspace{2mm}
\bm{F}(\alpha_{k},\tau_{k})^{T}\\
\bm{F}(\alpha_{k},\omega_{k})^{T}
\end{bmatrix}^{T}
\right\}^{-1}
\end{align}

By calculating the reduced CRB matrix, we can obtain that
\begin{align}
 &\bm{C}(\theta_{k},\tau_{k})=0,\\
 &\bm{C}(\theta_{k},\omega_{k})=0,
\end{align}

\begin{equation}
\bm{C}(\theta_{k},\theta_{k})=\left[\!\bm{F}(\theta_{k},\theta_{k})\!-\!\frac{A^{\dag}_{\alpha_{k}\theta_{k}}\cdot A_{\alpha_{k}\theta_{k}}}{2(\bm{y}^{\dag}_{rad}\bm{\Gamma}^{-1}\bm{y}_{rad})(\bm{a}^{\dag}(\theta_{k})\bm{\Lambda}^{-1}\bm{a}(\theta_{k}))}\!\right]^{-1}
\end{equation}

\begin{align}\label{12323}
\text{CRB}_{\tau_{k},\omega_{k}}&=
\begin{bmatrix}
\bm{C}(\tau_{k},\tau_{k})&\bm{C}(\tau_{k},\omega_{k})\notag\\
\bm{C}(\tau_{k},\omega_{k})^{T}&\bm{C}(\omega_{k},\omega_{k})
\end{bmatrix}\\
&=
\begin{bmatrix}
\bm{F}'(\tau_{k},\tau_{k})&\bm{F}'(\tau_{k},\omega_{k})\\
\bm{F}'(\tau_{k},\omega_{k})^{T}&\bm{F}'(\omega_{k},\omega_{k})
\end{bmatrix}
^{-1}
\end{align}

where $\bm{F}'$ denotes the reduced FIM on $\tau_{k}$ and $\omega_{k}$.
\begin{equation}
\bm{F}'(\tau_{k},\tau_{k})=\bm{F}(\tau_{k},\tau_{k})-\frac{A^{\dag}_{\alpha_{k}\tau_{k}}\cdot A_{\alpha_{k}\tau_{k}}}{2\bm{y}^{\dag}_{rad}\bm{\Gamma}^{-1}\bm{y}_{rad}\bm{a}^{\dag}(\theta_{k})\bm{\Lambda}^{-1}\bm{a}(\theta_{k})}
\end{equation}

\begin{equation}
\bm{F}'(\tau_{k},\omega_{k})=\bm{F}(\tau_{k},\omega_{k})-\frac{\text{Re}\{A^{\dag}_{\alpha_{k}\tau_{k}}\cdot A_{\alpha_{k}\omega_{k}}\}}{2\bm{y}^{\dag}_{rad}\bm{\Gamma}^{-1}\bm{y}_{rad}\bm{a}^{\dag}(\theta_{k})\bm{\Lambda}^{-1}\bm{a}(\theta_{k})}
\end{equation}

\begin{equation}
\bm{F}'(\omega_{k},\omega_{k})=\bm{F}(\omega_{k},\omega_{k})-\frac{A^{\dag}_{\alpha_{k}\omega_{k}}\cdot A_{\alpha_{k}\omega_{k}}}{2\bm{y}^{\dag}_{rad}\bm{\Gamma}^{-1}\bm{y}_{rad}\bm{a}^{\dag}(\theta_{k})\bm{\Lambda}^{-1}\bm{a}(\theta_{k})}
\end{equation}

where
\begin{equation}
A_{\alpha_{k}\theta_{k}}=2\alpha_{k}\bm{y}^{\dag}_{rad}\bm{\Gamma}^{-1}\bm{y}_{rad}\bm{a}^{\dag}(\theta_{k})\bm{\Lambda}^{-1}\frac{\partial\bm{a}(\theta_{k})}{\partial\theta_{k}}
\end{equation}

\begin{equation}
A_{\alpha_{k}\tau_{k}}=2\alpha_{k}\bm{y}^{\dag}_{rad}\bm{\Gamma}^{-1}\frac{\partial\bm{y}_{rad}}{\partial\tau_{k}}\bm{a}^{\dag}(\theta_{k})\bm{\Lambda}^{-1}\bm{a}(\theta_{k})
\end{equation}

\begin{equation}
A_{\alpha_{k}\omega_{k}}=2\alpha_{k}\bm{y}^{\dag}_{rad}\bm{\Gamma}^{-1}\frac{\partial\bm{y}_{rad}}{\partial\omega_{k}}\bm{a}^{\dag}(\theta_{k})\bm{\Lambda}^{-1}\bm{a}(\theta_{k})
\end{equation}

By further simplification, we have \emph{Theorem 1}.

\section{Proof of Theorem 2}

Based on the LFM signals expressed in (\ref{74664632}), we derive the exact  $\bm{F}'(\tau_{k},\tau_{k})$, $\bm{F}'(\tau_{k},\omega_{k})$ and $\bm{F}'(\omega_{k},\omega_{k})$, based on which the CRB can be obtained,
\begin{align}
\bm{F}'(\tau_{k},\tau_{k})=&\frac{2N|\alpha_{k}|^2}{\bar{\sigma}_{z}^{2}}\cdot f_{B}\frac{P_{rad}}{M}\sum_{i=1}^{M}\int_{0}^{T_{c}}|\frac{\partial s_{i}(t)}{\partial t}|^2dt\notag\\
=&\frac{2N|\alpha_{k}|^2}{\bar{\sigma}_{z}^{2}}\cdot f_{B}\frac{P_{rad}}{M}\sum_{i=1}^{M}\int_{-\frac{f_{B}}{2}}^{\frac{f_{B}}{2}}\frac{4\pi^{2}\delta T_{c}}{T_{0}}f^2|S_{i}(f)|^2df\notag\\
=&\frac{2N|\alpha_{k}|^2}{\bar{\sigma}_{z}^{2}}\cdot\frac{P_{rad}f_{B}^{3}\pi^{2}\delta T_{c}}{3}
\end{align}

\begin{align}
\bm{F}'(\tau_{k},\omega_{k})=&\frac{2N|\alpha_{k}|^2}{\bar{\sigma}_{z}^{2}}\cdot\text{Im}\left\{f_{B}\frac{P_{rad}}{M}\cdot\right.\notag\\
&\left.\sum_{i=1}^{M}\sum_{k=0}^{K-1}\int_{kT_{R}}^{kT_{R}+T_{0}}-\frac{4\pi f_{B}j}{T_{0}}t(t-kT_{R}-\frac{1}{2}T_{0})dt\right\}\notag\\
=&\frac{2N|\alpha_{k}|^2}{\bar{\sigma}_{z}^{2}}\cdot\text{Im}\left\{-\frac{P_{rad}f^{2}_{B}\delta\pi T_{0}T_{c}j}{3}\right\}
\end{align}

\begin{align}
\bm{F}'(\omega_{k},\omega_{k})=&\frac{2N|\alpha_{k}|^2}{\bar{\sigma}_{z}^{2}}\cdot\left[f_{B}\frac{P_{rad}}{M}\sum_{i=1}^{M}\sum_{k=0}^{K-1}\int_{kT_{R}}^{kT_{R}+T_{0}}t^2dt\right.\notag\\
&\left.-\frac{1}{\xi_{r}}\left(f_{B}\frac{P_{rad}}{M}\sum_{i=1}^{M}\sum_{k=0}^{K-1}\int_{kT_{R}}^{kT_{R}+T_{0}}tdt\right)\right]\notag\\
=&\frac{2N|\alpha_{k}|^2}{\bar{\sigma}_{z}^{2}}\cdot f_{B}P_{rad}\left[\sum_{k=0}^{K-1}\frac{3kT_{R}T^2_{0}+3k^2T^2_{R}T_{0}+T_{0}^3}{3}\right.\notag\\
&\left.-(\sum_{k=0}^{K-1}kT_{R}T_{0}+\frac{1}{2}kT_{0}^2)^2\right]\notag\\
=&\frac{2N|\alpha_{k}|^2}{\bar{\sigma}_{z}^{2}}\cdot \frac{P_{rad}f_{B} T_{c}}{12}\left[(\frac{\delta^2T_{c}^{2}}{T_{0}^2}-1)T_{R}^{2}+T_{0}^2\right]
\end{align}

Substituting $\bm{F}'(\tau_{k},\tau_{k})$, $\bm{F}'(\tau_{k},\omega_{k})$ and $\bm{F}'(\omega_{k},\omega_{k})$ into (\ref{12323}), we can have that
\begin{equation}
\text{CRB}(\tau_{k},\tau_{k})=\frac{\bar{\sigma}_{z}^2}{2N|\alpha_{k}|^2P_{rad}f_{B}\delta T_{C}}\!\cdot\!\frac{3[(K^2-1)T^2_{R}+T_{0}^2]}{f_{B}^2\pi^2[(K^2-1)T^{2}_{R}-3T_{0}^2]}
\end{equation}

\begin{equation}
\text{CRB}(\omega_{k},\omega_{k})=\frac{\bar{\sigma}_{z}^2}{2N|\alpha_{k}|^2P_{rad}f_{B}\delta T_{C}}\cdot\frac{12}{(K^2-1)T^{2}_{R}-3T_{0^2}}
\end{equation}

\bibliographystyle{IEEEtran}
\bibliography{2021-JCR-final-singlecol}

\begin{thebibliography}{10}
\providecommand{\url}[1]{#1}
\csname url@samestyle\endcsname
\providecommand{\newblock}{\relax}
\providecommand{\bibinfo}[2]{#2}
\providecommand{\BIBentrySTDinterwordspacing}{\spaceskip=0pt\relax}
\providecommand{\BIBentryALTinterwordstretchfactor}{4}
\providecommand{\BIBentryALTinterwordspacing}{\spaceskip=\fontdimen2\font plus
\BIBentryALTinterwordstretchfactor\fontdimen3\font minus
  \fontdimen4\font\relax}
\providecommand{\BIBforeignlanguage}[2]{{%
\expandafter\ifx\csname l@#1\endcsname\relax
\typeout{** WARNING: IEEEtran.bst: No hyphenation pattern has been}%
\typeout{** loaded for the language `#1'. Using the pattern for}%
\typeout{** the default language instead.}%
\else
\language=\csname l@#1\endcsname
\fi
#2}}
\providecommand{\BIBdecl}{\relax}
\BIBdecl

\bibitem{ZHOU2020253}
Y.~Zhou, L.~Liu, L.~Wang, N.~Hui, X.~Cui, J.~Wu, Y.~Peng, Y.~Qi, and C.~Xing,
  ``Service-aware 6g: An intelligent and open network based on the convergence
  of communication, computing and caching,'' \emph{Digital Communications and
  Networks}, vol.~6, no.~3, pp. 253--260, 2020.

\bibitem{8663968}
Y.~Zhou, L.~Tian, L.~Liu, and Y.~Qi, ``Fog computing enabled future mobile
  communication networks: A convergence of communication and computing,''
  \emph{IEEE Communications Magazine}, vol.~57, no.~5, pp. 20--27, 2019.

\bibitem{Fan2020Survey}
F.~{Liu}, C.~{Masouros}, A.~P. {Petropulu}, H.~{Griffiths}, and L.~{Hanzo},
  ``Joint radar and communication design: Applications, state-of-the-art, and
  the road ahead,'' \emph{IEEE Transactions on Communications}, vol.~68, no.~6,
  pp. 3834--3862, 2020.

\bibitem{mahal2017spectral}
J.~A. Mahal, A.~Khawar, A.~Abdelhadi, and T.~C. Clancy, ``Spectral coexistence
  of mimo radar and mimo cellular system,'' \emph{IEEE Trans. Aerosp. Electron.
  Syst.}, vol.~53, no.~2, pp. 655--668, Apr. 2017.

\bibitem{biswas2020design}
S.~Biswas, K.~Singh, O.~Taghizadeh, and T.~Ratnarajah, ``Design and analysis of
  fd mimo cellular systems in coexistence with mimo radar,'' \emph{IEEE
  Transactions on Wireless Communications}, vol.~19, no.~7, pp. 4727--4743,
  2020.

\bibitem{8352726}
J.~{Qian}, M.~{Lops}, {Le Zheng}, X.~{Wang}, and Z.~{He}, ``Joint system design
  for coexistence of mimo radar and mimo communication,'' \emph{IEEE Trans.
  Signal Process.}, vol.~66, no.~13, pp. 3504--3519, Jul. 2018.

\bibitem{liu2017robust}
F.~Liu, C.~Masouros, A.~Li, and T.~Ratnarajah, ``Robust mimo beamforming for
  cellular and radar coexistence,'' \emph{IEEE Wireless Commun. Lett.}, vol.~6,
  no.~3, pp. 374--377, Jun. 2017.

\bibitem{8447442}
S.~{Biswas}, K.~{Singh}, O.~{Taghizadeh}, and T.~{Ratnarajah}, ``Coexistence of
  mimo radar and fd mimo cellular systems with qos considerations,'' \emph{IEEE
  Trans. Wireless Commun.}, vol.~17, no.~11, pp. 7281--7294, Nov. 2018.

\bibitem{8531782}
Q.~{He}, Z.~{Wang}, J.~{Hu}, and R.~S. {Blum}, ``Performance gains from
  cooperative mimo radar and mimo communication systems,'' \emph{IEEE Signal
  Process. Lett.}, vol.~26, no.~1, pp. 194--198, Jan. 2019.

\bibitem{Riihonen2021Opt}
S.~D. Liyanaarachchi, T.~Riihonen, C.~B. Barneto, and M.~Valkama, ``Optimized
  waveforms for 5g–6g communication with sensing: Theory, simulations and
  experiments,'' \emph{IEEE Transactions on Wireless Communications}, pp. 1--1,
  2021.

\bibitem{chiriyath2016inner}
A.~R. Chiriyath, B.~Paul, G.~M. Jacyna, and D.~W. Bliss, ``Inner bounds on
  performance of radar and communications co-existence,'' \emph{IEEE Trans.
  Sig. Process.}, vol.~64, no.~2, pp. 464--474, Jan. 2016.

\bibitem{chiriyath2017radar}
A.~R. Chiriyath, B.~Paul, and D.~W. Bliss, ``Radar-communications convergence:
  Coexistence, cooperation, and co-design,'' \emph{IEEE Trans. Cogn. Commun.
  Network.}, vol.~3, no.~1, pp. 1--12, Mar. 2017.

\bibitem{chiriyath2017simultaneous}
------, ``Simultaneous radar detection and communications performance with
  clutter mitigation,'' in \emph{Proc. 2017 IEEE Radar Conference (RadarConf)},
  May 2017, pp. 0279--0284.

\bibitem{ahmadipour2021informationtheoretic}
M.~Ahmadipour, M.~Kobayashi, M.~Wigger, and G.~Caire, ``An
  information-theoretic approach to joint sensing and communication,'' 2021.

\bibitem{rong2017multiple}
Y.~Rong, A.~R. Chiriyath, and D.~W. Bliss, ``Multiple-antenna multiple-access
  joint radar and communications systems performance bounds,'' in \emph{Proc.
  2017 51st Asilomar Conf.}, Oct. 2017, pp. 1296--1300.

\bibitem{cheng2018outer}
C.~Li, N.~Raymondi, B.~Xia, and A.~Sabharwal, ``Outer bounds for mimo
  communicating radars: Three-node uplink,'' in \emph{Proc. 2018 IEEE Asilomar
  Conf.}, Oct. 2018, pp. 934--938.

\bibitem{liu2018mumimo}
F.~Liu, C.~Masouros, A.~Li, H.~Sun, and L.~Hanzo, ``Mu-mimo communications with
  mimo radar: From co-existence to joint transmission,'' \emph{IEEE Trans.
  Wireless Commun.}, vol.~17, no.~4, pp. 2755--2770, Apr. 2018.

\bibitem{liu2020radar}
F.~Liu, W.~Yuan, C.~Masouros, and J.~Yuan, ``Radar-assisted predictive
  beamforming for vehicular links: Communication served by sensing,''
  \emph{IEEE Transactions on Wireless Communications}, vol.~19, no.~11, pp.
  7704--7719, 2020.

\bibitem{Dokhanchi2019mmWave}
S.~H. Dokhanchi, B.~S. Mysore, K.~V. Mishra, and B.~Ottersten, ``A mmwave
  automotive joint radar-communications system,'' \emph{IEEE Transactions on
  Aerospace and Electronic Systems}, vol.~55, no.~3, pp. 1241--1260, 2019.

\bibitem{wang2019power}
F.~Wang, H.~Li, and M.~A. Govoni, ``Power allocation and co-design of
  multicarrier communication and radar systems for spectral coexistence,''
  \emph{IEEE Transactions on Signal Processing}, vol.~67, no.~14, pp.
  3818--3831, 2019.

\bibitem{xu2021wideband}
Z.~Xu and A.~Petropulu, ``A wideband dual function radar communication system
  with sparse array and ofdm waveforms,'' 2021.

\bibitem{kumari2018ieee802}
P.~Kumari, J.~Choi, N.~Gonz\'{a}lez-Prelcic, and R.~W. Heath, ``Ieee
  802.11ad-based radar: An approach to joint vehicular communication-radar
  system,'' \emph{IEEE Trans. Veh. Technol.}, vol.~67, no.~4, pp. 3012--3027,
  Apr. 2018.

\bibitem{Kumari2020Adaptive}
P.~Kumari, S.~A. Vorobyov, and R.~W. Heath, ``Adaptive virtual waveform design
  for millimeter-wave joint communication–radar,'' \emph{IEEE Transactions on
  Signal Processing}, vol.~68, pp. 715--730, 2020.

\bibitem{barneto2021full}
C.~B. Barneto, S.~D. Liyanaarachchi, M.~Heino, T.~Riihonen, and M.~Valkama,
  ``Full duplex radio/radar technology: The enabler for advanced joint
  communication and sensing,'' \emph{IEEE Wireless Communications}, vol.~28,
  no.~1, pp. 82--88, 2021.

\bibitem{barneto2021FDmag}
------, ``Full duplex radio/radar technology: The enabler for advanced joint
  communication and sensing,'' \emph{IEEE Wireless Communications}, vol.~28,
  no.~1, pp. 82--88, 2021.

\bibitem{4350230}
J.~{Li} and P.~{Stoica}, ``Mimo radar with colocated antennas,'' \emph{IEEE
  Sig. Process. Mag.}, vol.~24, no.~5, pp. 106--114, Sept. 2007.

\bibitem{6571320}
J.~{Bai} and A.~{Sabharwal}, ``Distributed full-duplex via wireless
  side-channels: Bounds and protocols,'' \emph{IEEE Trans. Wireless Commun.},
  vol.~12, no.~8, pp. 4162--4173, Aug. 2013.

\bibitem{chiriyath2015joint}
A.~R. Chiriyath and D.~W. Bliss, ``Joint radar-communications performance
  bounds: Data versus estimation information rates,'' in \emph{Proc. 2015 IEEE
  Military Commun. Conf.}, Oct. 2015, pp. 1491--1496.

\bibitem{4156404}
N.~H. {Lehmann}, E.~{Fishler}, A.~M. {Haimovich}, R.~S. {Blum}, D.~{Chizhik},
  L.~J. {Cimini}, and R.~A. {Valenzuela}, ``Evaluation of transmit diversity in
  mimo-radar direction finding,'' \emph{IEEE Trans. Sig. Process.}, vol.~55,
  no.~5, pp. 2215--2225, May 2007.

\bibitem{4801449}
K.~S. {Ahn} and R.~W. {Heath}, ``Performance analysis of maximum ratio
  combining with imperfect channel estimation in the presence of cochannel
  interferences,'' \emph{IEEE Trans. Wireless Commun.}, vol.~8, no.~3, pp.
  1080--1085, Mar. 2009.

\bibitem{cover2012elements}
T.~M. Cover and J.~A. Thomas, \emph{Elements of information theory}.\hskip 1em
  plus 0.5em minus 0.4em\relax John Wiley \& Sons, 2012.

\bibitem{duarte2012experiment}
M.~Duarte, C.~Dick, and A.~Sabharwal, ``Experiment-driven characterization of
  full-duplex wireless systems,'' \emph{IEEE Trans. Wireless Commun.}, vol.~11,
  no.~12, pp. 4296--4307, Dec. 2012.

\bibitem{wang2014radar}
W.-Q. Wang and H.~Shao, ``Radar-to-radar interference suppression for
  distributed radar sensor networks,'' \emph{Remote Sensing}, vol.~6, no.~1,
  pp. 740--755, 2014.

\bibitem{su2018time}
J.~Su, H.~Tao, M.~Tao, J.~Xie, Y.~Wang, and L.~Wang, ``Time-varying sar
  interference suppression based on delay-doppler iterative decomposition
  algorithm,'' \emph{Remote Sensing}, vol.~10, no.~9, p. 1491, 2018.

\bibitem{7801128}
B.~Xia, C.~Li, and Q.~Jiang, ``Outage performance analysis of multi-user
  selection for two-way full-duplex relay systems,'' \emph{IEEE Commun. Lett.},
  vol.~21, no.~4, pp. 933--936, Apr. 2017.

\bibitem{bliss2013adaptive}
D.~W. Bliss and S.~Govindasamy, \emph{Adaptive wireless communications: MIMO
  channels and networks}.\hskip 1em plus 0.5em minus 0.4em\relax Cambridge
  University Press, 2013.

\bibitem{kay1993fundamental}
S.~K. Kay, \emph{Fundamentals of Statistical Signal Processing: Estimation
  Theory}, 5th~ed.\hskip 1em plus 0.5em minus 0.4em\relax Prentice Hall, 1993.

\bibitem{wittman2012fisher}
\BIBentryALTinterwordspacing
D.~Wittman, ``Fisher matrix for beginners.'' [Online]. Available:
  \url{http://wittman.physics.ucdavis.edu/Fisher-matrix-guide.pdf}
\BIBentrySTDinterwordspacing

\bibitem{kocjancic2018}
L.~Kocjancic, A.~Balleri, and T.~Merlet, ``Multibeam radar based on linear
  frequency modulated waveform diversity,'' \emph{IET Radar, Sonar \& Navig.},
  vol.~12, no.~11, pp. 1320--1329, 2018.

\bibitem{rao2013}
M.~H. Rao, G.~Sharma, and K.~R. Rajeswari, ``Orthogonal phase coded waveforms
  for mimo radars,'' \emph{Int. J. Comp. Appl.}, vol.~63, no.~6, 2013.

\bibitem{8486331}
Z.~{Zhang}, Y.~{Sun}, A.~{Sabharwal}, and Z.~{Chen}, ``Impact of channel state
  misreporting on multi-user massive mimo scheduling performance,'' in
  \emph{Proc. IEEE INFOCOM}, Apr. 2018, pp. 917--925.

\bibitem{tsai2011path}
M.~Tsai, ``Path-loss and shadowing (large-scale fading),'' Nat. Taiwan Univ.,
  Tech. Rep., 2011.

\bibitem{dogandzic2001}
A.~Dogandzic and A.~Nehorai, ``Cramer-rao bounds for estimating range,
  velocity, and direction with an active array,'' \emph{IEEE Trans. Signal
  Process.}, vol.~49, no.~6, pp. 1122--1137, Jun. 2001.

\end{thebibliography}
\end{document}